\documentclass[fleqn,english,twocolumn, secthm, amsthm]{autart}
\usepackage[T1]{fontenc}
\usepackage[latin9]{inputenc}
\usepackage{booktabs}
\usepackage{amsmath}
\usepackage{amssymb}
\usepackage{graphicx}

\makeatletter

\providecommand{\tabularnewline}{\\}

\usepackage[dcucite]{harvard}
\usepackage{balance}
\usepackage{dsfont}
\newcommand{\mathbbm}[1]{\mathds{#1}}

\usepackage{times} 

\usepackage{url}

\usepackage{url}
\usepackage{verbatim}

\usepackage{color}

\makeatother

\usepackage{babel}
\begin{document}
\begin{frontmatter} 
\title{Randomized Transmission Protocols for Protection against  \\ Jamming Attacks in Multi-Agent Consensus\thanksref{footnoteinfo}} \thanks[footnoteinfo]{This work was supported  in part by JST ERATO HASUO Metamathematics for Systems Design Project (No.\ JPMJER1603), by  JST CREST Grant  No.\ JPMJCR15K3, and by JSPS under Grant-in-Aid for Scientific Research Grant No.~18H04020. The material in this paper was  presented at the 56th IEEE Conference on Decision and Control, 2017, Melbourne, Australia.}
\author[nii1]{Ahmet Cetinkaya} \ead{cetinkaya@nii.ac.jp},
\author[titech1]{Kaito Kikuchi} \ead{kikuchi@dsl.mei.titech.ac.jp},
\author[titech1]{Tomohisa Hayakawa} \ead{hayakawa@sc.e.titech.ac.jp},
\author[titech2]{Hideaki Ishii} \ead{ishii@c.titech.ac.jp}
\address[nii1]{Information Systems Architecture Science Research Division, National Institute of Informatics, Tokyo, 101-8430, Japan}
\address[titech1]{Department of Systems and Control Engineering, Tokyo Institute of Technology, Tokyo 152-8552, Japan} 
\address[titech2]{Department of Computer Science, Tokyo Institute of Technology, Yokohama, 226-8502, Japan}

\begin{abstract} Multi-agent consensus under jamming attacks is investigated.
Specifically, inter-agent communications over a network are assumed
to fail at certain times due to jamming of transmissions by a malicious
attacker. A new stochastic communication protocol is proposed to achieve
finite-time practical consensus between agents. In this protocol,
communication attempt times of agents are randomized and unknown by
the attacker until after the agents make their communication attempts.
Through a probabilistic analysis, we show that the proposed communication
protocol, when combined with a stochastic ternary control law, allows
agents to achieve consensus regardless of the frequency of attacks.
We demonstrate the efficacy of our results by considering two different
strategies of the jamming attacker: a deterministic attack strategy
and a more malicious communication-aware attack strategy. \end{abstract}\end{frontmatter}

\section{Introduction }

Nowadays, control systems heavily utilize information and communication
technologies. Especially, the Internet of Things is becoming widespread
and remote sensing/control operations can now take place over wireless
networks. With these new developments, the risk of cyber attacks against
control systems is also increasing. Communication channels used in
control systems are vulnerable to cyber attacks and ensuring cyber
security of control systems has become a very important challenge
\cite{wholejournal2015new}.

Networked control systems are threatened by different types of cyber
attacks. For instance, on a vulnerable network, measurement and control
data can be altered by a malicious attacker \cite{fawzi2014secure}.
In certain situations, attackers can even inject false data into the
system without being noticed \cite{mo2010false}. These attacks require
the attacker to be knowledgeable about the system dynamics. In the
context of multi-agent systems, the presence of faulty or even malicious
agents not following the given protocols may affect the global behavior
of the overall system. There is a rich history in computer science
on the development of resilient consensus algorithms (e.g., \citeasnoun{lynch1996distributed}, \citeasnoun{azadmanesh2002asynchronous}).
Recently, this problem has gained interest in systems and control
as well \cite{leblanc2013resilient,tseng2015fault,dibaji2016resilient,dibaji2017resilient}. 

On the other hand, attackers who have limited information about the
control system can resort to denial-of-service (DoS) attacks to prevent
communication over networks. For instance, malicious routers in a
network may intentionally drop measurement and control data \cite{awerbuch2008odsbr,mahmoud2014security}.
Moreover, denial-of-service on wireless networks can also happen in
the form of jamming attacks. A jamming attacker can block the transmissions
on a wireless channel by emitting strong interference signals \cite{xu2005feasibility,pelechrinis2011}.
Recently, researchers explored the effect of jamming and other types
of denial-of-service attacks on networked control systems \cite{de_persis_2016,shishehsiam2016,cetinkaya2016tac,feng2017resilient,cetinkayatactoappear2018,ahmetsiam2018}.
Moreover, the effect of jamming on multi-agent consensus has also
been explored \cite{senejohnny_2015,senejohnny2017jamming}. 

One of the main challenges in studying the multi-agent consensus problem
under jamming attacks is that the attacker's actions cannot be known
a priori. To account for the uncertainty in the generation of attacks,
the works \cite{senejohnny_2015,senejohnny2017jamming} characterized
jamming attacks through their average duration and frequency. It is
shown there that multi-agent consensus can be achieved if the duration
and the frequency of attacks satisfy certain conditions. Specifically,
these works consider a self-triggered control framework, where each
agent attempts to communicate with its neighbors and update its local
control input only when a triggering condition is satisfied. For consensus,
it is required that the ratio of the duration of the attacks to the
total time is less than one. This ensures that the jamming does not
span the entire time. Note that under the self-triggering framework,
the communication attempt times for the agents are deterministic.
Thus, an attacker who is knowledgeable on the multi-agent system can
determine those time instants. This allows the attacker to block the
communication by turning on the jamming attack very briefly at those
instants without violating the duration condition. To avoid this issue,
a restriction on the attack frequency becomes necessary. Specifically,
the frequency of the attacks is required to be less than the frequency
of the communication attempts by the agents. 

Motivated by the discussion above, our goal in this paper is to investigate
attack scenarios where the jamming is turned on and off very frequently.
Our main contribution is a new stochastic consensus framework to deal
with those attack scenarios. In our framework, we use the ternary
control laws previously used in \cite{de_persis_ternary_2013,senejohnny_2015,senejohnny2017jamming}.
However, instead of the self-triggering method utilized in those works,
we propose a stochastic communication protocol that can achieve consensus
regardless of the frequency of the attacks. In this protocol, each
agent attempts to communicate with its neighbors at random time instants.
These time instants are hence unknown by the attacker. 

We consider two attack strategies that are restricted by their average
duration but \emph{not }by their frequency. In the first strategy,
the starting time and the duration of the jamming attacks are deterministic
and do not depend on whether the agents try to communicate. On the
other hand, in the second strategy the attacker is aware of the communication
attempts of the agents and can preserve energy by turning off jamming
right after a communication attempt is blocked. We show that in both
strategies, our proposed stochastic communication protocol guarantees
infinitely many successful communications in the long run. Furthermore,
by using a probabilistic analysis, we show that almost-sure finite-time
practical consensus is achieved regardless of attack frequency as
long as the average ratio of attack durations is less than hundred-percent. 

Our approach for analyzing the consensus under jamming differs largely
from those in the literature. In particular, for the deterministic
communication strategy proposed in \cite{senejohnny_2015,senejohnny2017jamming},
bounds on attack frequency can be used for establishing an upper-bound
for the interval between two consecutive successful communication
times of an agent. Here in this paper, such an upper-bound is not
available and there is a positive probability that any finite number
of consecutive communication attempts can be blocked by a jamming
attacker. This difference is due to the fact that we do not consider
a bound for attack frequencies and our communication protocol involves
randomization of transmission times. We also note that although there
are several works that deal with random connectivity issues and randomly
switching graph topologies in multi-agent systems (e.g., \citeasnoun{tahbaz2010consensus}, \citeasnoun{zhang2010consensus}, \citeasnoun{you2013consensus}),
the analysis techniques in this paper are completely different from
those works due to our approach of intentionally randomizing the inter-agent
communication times to mitigate jamming attacks which occur at uncertain
times. 

Our analysis for consensus relies on first establishing that under
randomized transmissions, all agents can communicate with their neighbors
infinitely many times in the long run. This is shown for the deterministic
and the communication-aware attacks using different techniques. In
the case of deterministic attacks, the independence of attacks and
communication attempts plays an important role. Another big role is
played by the uniform distribution of random communication attempt
times. On the other hand, in the case of communication-aware attacks,
the timing of attacks depends on all previous history of the communication
times of agents. In the analysis of this case, we construct a filtration
that represents the progression of the actions of the agents and those
of the attacker. By utilizing this filtration, we show that our protocol
can achieve a positive probability of at least one successful inter-agent
transmission during carefully selected sufficiently long intervals
spanning the time domain. We then utilize the monotone-convergence
theorem for sets to show that even in communication-aware attacks,
each agent can make infinitely many successful communications in the
long run. This result allows us to show that with suitable choice
of control parameters, each agent would be able to select appropriate
control actions and apply them long enough to reach consensus in finite
time.

In this paper, we show that randomization in inter-agent communications
enables agents to reach consensus regardless of the frequency of jamming
attacks. In recent works, randomization in communication has been
exploited in different ways. For instance, randomized gossip algorithms
is used in \citeasnoun{boyd2006randomized} to allow networked operation
under limited computation and communication resources. Furthermore,
the work by \citeasnoun{dibaji2016resilient} introduced randomness
in quantization as well as in communication times to increase resiliency
against malicious nodes in multi-agent systems. Such advantages of
using probabilistic methods have been found in resilient consensus
in computer science and are often referred to as \textquotedbl impossibility
results\textquotedbl{} (e.g., \citeasnoun{lynch1996distributed}).
In addition, random frequency hopping techniques are utilized by \citeasnoun{navda2007using}
and \citeasnoun{popper2010anti} to mitigate jamming in wireless networks. 

The paper is organized as follows. In Section~\ref{sec:Multi-Agent-Consensus-Problem},
we explain the multi-agent consensus problem under jamming attacks.
We propose a stochastic communication protocol and provide conditions
for consensus under jamming attacks in Section~\ref{sec:Stochastic-Communication-Protoco}.
Then we discuss our protocol's efficacy under deterministic and communication-aware
attacks in Section~\ref{sec:Deterministic-Jamming-and}. In Section~\ref{sec:Numerical-Examples},
we present numerical examples to demonstrate our results. Finally,
we conclude the paper in Section~\ref{sec:Conclusion}. 

We note that part of the results in Sections~\ref{sec:Stochastic-Communication-Protoco}
and \ref{sec:Deterministic-Jamming-and} appeared in our preliminary
report \cite{kikuchicdc2017} without proofs. In this paper, we provide
complete proofs and more detailed discussions in Sections~\ref{sec:Stochastic-Communication-Protoco}
and \ref{sec:Deterministic-Jamming-and}. Furthermore, new numerical
examples are presented in Section~\ref{sec:Numerical-Examples}. 

The notation used in the paper is fairly standard. Specifically, we
denote positive and nonnegative integers by $\mathbb{N}$ and $\mathbb{N}_{0}$,
respectively. Furthermore, we use $(\cdot)^{\mathrm{T}}$ to denote
transpose, $|S|$ to denote the Lebesgue measure of a set $S\subset\mathbb{R}$,
and $A\setminus B$ to denote the set of elements that belong to set
$A$, but not to set $B$. The notations $\mathrm{\mathbb{P}}[\cdot]$
and $\mathrm{\mathbb{E}}[\cdot]$ respectively denote the probability
and the expectation on a probability space $(\Omega,\mathcal{F},\mathbb{P})$.
Moreover, we use $\mathbbm{1}[E]:\Omega\to\{0,1\}$ for the indicator
of the event $E\in\mathcal{F}$, that is, $\mathbbm{1}[E](\omega)=1$,
$\omega\in E$, and $\mathbbm{1}[E](\omega)=0$, $\omega\notin E$.
To simplify the presentation, we omit the $\omega\in\Omega$ in the
notation of random variables in certain equations. 

\section{Multi-Agent Consensus Under Jamming Attacks \label{sec:Multi-Agent-Consensus-Problem}}

In this paper, we investigate the consensus problem for a multi-agent
system composed of $n$ agents with scalar dynamics. The communication
topology of the multi-agent system is represented by an undirected
connected graph $\mathcal{G}=(\mathcal{V},\mathcal{E})$, where $\mathcal{V}=\{1,\ldots,n\}$
represents the set of nodes corresponding to the $n$ agents, and
$\mathcal{E}\subset\mathcal{V}\times\mathcal{V}$ is the set of edges
corresponding to the communication links between the agents. Let $\mathcal{N}_{i}$
be the set of neighbors and $d^{i}$ be the degree of node $i$. We
use $L\in\mathbb{R}^{n\times n}$ to denote the Laplacian matrix associated
with $\mathcal{G}$. Note that $L$ is a symmetric matrix since $\mathcal{G}$
is an undirected graph. 

The evolution of the states of the multi-agent system is characterized
through the scalar dynamics
\begin{align}
\dot{x}^{i}(t) & =u^{i}(t),\quad t\geq0,\label{eq:scalar-dynamics}
\end{align}
 where $x^{i}(t)$ and $u^{i}(t)$ respectively denote the state and
the local control input for agent $i$. 

In this paper we design a piecewise-constant control input $u^{i}(t)$
for each agent $i$, as well as a protocol for the communication between
the agents so that the agents achieve practical consensus, that is,
$x(t)\triangleq[\begin{array}{cccc}
x^{1}(t) & x^{2}(t) & \cdots & x^{n}(t)\end{array}]^{\mathrm{T}}$ representing the agent states converges in finite time to a vector
$x^{*}\in\mathbb{R}^{n}$ belonging to the approximate consensus set
$\mathcal{D}_{\varepsilon}$ with $\varepsilon>0$ given by 
\begin{equation}
\mathcal{D}_{\varepsilon}\triangleq\Bigl\{ x\in\mathbb{R}^{n}:\Big|\sum_{j\in\mathcal{N}_{i}}(x^{j}-x^{i})\Big|<\varepsilon,\hspace{5pt}i\in\mathcal{V}\Bigr\}.\label{eq:finalset}
\end{equation}

In what follows we first discuss the jamming attacks on the communication
channels between the agents, and then we explain our proposed control
and communication protocols for achieving consensus under jamming. 

\subsection{Jamming Attacks \label{subsec:Jamming-Attacks}}

We consider scenarios where the communication channels between the
agents are disabled by a jamming attacker. Specifically, we assume
that when there is a jamming attack, communication on all links $\mathcal{E}$
fail. This setup allows us to model attacks on a shared network that
the agents use for communication. We note that the results presented
in the paper can also be extended with small modifications to more
general cases where the communication links are attacked individually
by multiple attackers. 

We follow the approach of \citeasnoun{de_persis_2016}, \citeasnoun{senejohnny_2015},
\citeasnoun{senejohnny2017jamming} and use the sequences $\{a_{k}\geq0\}_{k\in\mathbb{N}_{0}}$
and $\{\tau_{k}\geq0\}_{k\in\mathbb{N}_{0}}$ to characterize the
starting time instants and durations of the sequence of attacks, respectively.
Specifically, $a_{k}$ represents the starting time instant for the
$k$th attack, and $\tau_{k}$ represents its duration. For each $k\in\mathbb{N}_{0}$,
these scalars are assumed to satisfy $a_{k+1}>a_{k}+\tau_{k}$ (see
the time sequence at the bottom of Figure~\ref{figure:attack_pattern}). 

\begin{figure}
\begin{centering}
\includegraphics[width=0.8\columnwidth]{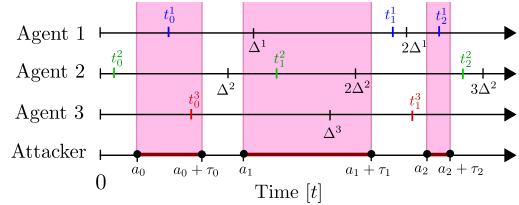} \vskip -9pt
\par\end{centering}
\caption{Communication attempt times of $3$ agents represented with $\{t_{k}^{i}\}_{k\in\mathbb{N}_{0}}$,
$i\in\{1,2,3\}$, together with the jamming attacks represented with
$\{a_{k}\}_{k\in\mathbb{N}_{0}}$ and $\{\tau_{k}\}_{k\in\mathbb{N}_{0}}$. }
\label{figure:attack_pattern} \vskip -2pt
\end{figure}

Now, let $\mathcal{A}_{k}\triangleq[a_{k},a_{k}+\tau_{k}]$ denote
the $k$th attack interval (or instant if $\tau_{k}=0$), during which
the attacker prevents all transmissions on the communication channels
between the agents. For any time interval $[\tau,t]\subset[0,\infty)$,
we use $\overline{\mathcal{A}}(\tau,t)$ to denote the set of times
under jamming attacks, that is, 
\begin{equation}
\overline{\mathcal{A}}(\tau,t)\triangleq\cup_{k\in\mathbb{N}_{0}}\mathcal{A}_{k}\cap[\tau,t].\label{eq:Xi-def}
\end{equation}
 Furthermore, for the same interval $[\tau,t]$, $\overline{\mathcal{A}}^{\mathrm{c}}(\tau,t)$
denotes the complement of $\overline{\mathcal{A}}(\tau,t)$ in the
sense that it is the set of times where there is no attack, that is,
\begin{equation}
\overline{\mathcal{A}}^{\mathrm{c}}(\tau,t)\triangleq[\tau,t]\setminus\overline{\mathcal{A}}(\tau,t).\label{eq:Theta-def}
\end{equation}

Conducting jamming attacks requires energy for transmitting interfering
radio signals \cite{xu2005feasibility}. Thus, an attacker with limited
resources would not be able to continuously jam the communication
channels for a long time. In such cases, the attacker may repeat cycles
of jamming and idling to preserve energy. The following assumption
from \citeasnoun{senejohnny_2015} provides a characterization of
the duration of jamming for various attack scenarios. 

\begin{assum} \label{Assumption:duration_sync} There exist $\kappa\geq0$
and $\rho\in(0,1)$ such that for each $\tau\geq0$ and $t\geq\tau$,
\begin{equation}
|\overline{\mathcal{A}}(\tau,t)|\leq\kappa+\rho(t-\tau),\label{eq:attack-duration-condition}
\end{equation}
 where $|\overline{\mathcal{A}}(\tau,t)|$ represents the total duration
of the attacks in the interval $[\tau,t]$. \end{assum}

Notice that \eqref{eq:attack-duration-condition} implies $\limsup_{t\to\infty}|\overline{\mathcal{A}}(0,t)|/t\leq\rho$.
As a consequence, the scalar $\rho\in(0,1)$ can be considered as
an upper-bound on the ratio of the duration of attacks in long intervals,
and it is related to the average energy used by the attacker. Under
Assumption~\ref{Assumption:duration_sync}, jamming attacks are allowed
to start at arbitrary time instants as long as \eqref{eq:attack-duration-condition}
holds. Note also that the longest duration for continuous jamming
allowed under Assumption~\ref{Assumption:duration_sync} can be obtained
as $\kappa/(1-\rho)$. Here, the scalar $\kappa$ can be selected
to model the attacker's capabilities for continuous jamming. 

\subsection{Stochastic Ternary Control}

To achieve consensus we employ the ternary control approach previously
used in \citeasnoun{de_persis_ternary_2013}, \citeasnoun{senejohnny_2015},
\citeasnoun{senejohnny2017jamming}. However, instead of the self-triggering
method utilized in those studies, we propose a stochastic communication
protocol. In what follows, we first explain the control framework.
We then discuss our communication protocol in detail in Section~\ref{sec:Stochastic-Communication-Protoco}. 

Each agent $i\in\mathcal{V}$ attempts communicating with its neighbors
$\mathcal{N}_{i}$ at times $t_{k}^{i}\geq0$, $k\in\mathbb{N}_{0}$.
In particular, at a communication attempt time $t_{k}^{i}$, agent
$i$ sends an information request to all of its neighbors and asks
for their states. If there is no jamming at time $t_{k}^{i}$, then
the neighbors of agent $i$ receive the request and send back their
states, which will be used in the update of agent $i$'s local control
input. In the case where there is a jamming attack at time $t_{k}^{i}$,
agent $i$ cannot send/receive information. 

We use $\varphi_{k}^{i}\in\{0,1\}$ to indicate whether the communication
attempt at time $t_{k}^{i}$ is successful or not. In particular,
$\varphi_{k}^{i}=0$ indicates a failure at time $t_{k}^{i}$ due
to a jamming attack, and $\varphi_{k}^{i}=1$ implies that agent $i$
successfully communicates with its neighbors at time $t_{k}^{i}$.
Observe that 
\begin{align}
\varphi_{k}^{i} & =\mathbbm{1}[t_{k}^{i}\in\overline{\mathcal{A}}^{\mathrm{c}}(0,t_{k}^{i})],\quad k\in\mathbb{N}_{0}.\label{eq:phi-i-k-def}
\end{align}
In this paper communication attempt times $t_{k}^{i}$ are random
variables, and consequently, $\varphi_{k}^{i}$ are also random variables. 

Now, let ${\rm ave}^{i}(t)\triangleq\sum_{j\in\mathcal{N}_{i}}(x^{j}(t)-x^{i}(t))$
and ${\rm sign}_{\varepsilon}$: ${\mathbb{R}\to\{-1,0,1\}}$ be defined
by 
\begin{align*}
{\rm sign}_{\varepsilon}(z) & \triangleq\begin{cases}
{\rm sign}(z), & {\rm if}\hspace{5pt}|z|\geq\varepsilon,\\
0, & {\rm otherwise,}
\end{cases}
\end{align*}
with $\varepsilon>0$ given in \eqref{eq:finalset}. 

In our framework, agent $i$ attempts to communicate with its neighbors
at time $t_{k}^{i}$ to compute ${\rm sign}_{\varepsilon}({\rm ave}^{i}(t_{k}^{i}))$.
Notice that ${\rm sign}_{\varepsilon}({\rm ave}^{i}(t_{k}^{i}))$
is a ternary indicator of the location of neighboring agents' states
with respect to agent $i$'s own state. Let 
\begin{align}
\tilde{u}_{k}^{i} & \triangleq\begin{cases}
{\rm sign}_{\varepsilon}({\rm ave}^{i}(t_{k}^{i})),\quad & \mathrm{if}\,\,\varphi_{k}^{i}=1,\\
0,\quad & \mathrm{otherwise},
\end{cases}\quad k\in\mathbb{N}_{0}.\label{eq:utilde-def}
\end{align}
Notice that $\tilde{u}_{k}^{i}=\varphi_{k}^{i}{\rm sign}_{\varepsilon}({\rm ave}^{i}(t_{k}^{i}))$.
Hence, $\tilde{u}_{k}^{i}\neq0$ implies $\varphi_{k}^{i}=1$ (i.e.,
agent $i$ successfully communicates with its neighbors at time $t_{k}^{i}$)
and ${\rm sign}_{\varepsilon}({\rm ave}^{i}(t_{k}^{i}))\in\{-1,1\}$
(i.e., $|{\rm ave}^{i}(t_{k}^{i})|\geq\varepsilon$). Now, let $\hat{u}_{k}^{i}\triangleq\tilde{u}_{k}^{i}$
for $k=0$ and
\begin{align}
\hat{u}_{k}^{i} & \triangleq\begin{cases}
\tilde{u}_{k}^{i}, & \mathrm{if}\,\,\tilde{u}_{k}^{i}\neq0,\\
\tilde{u}_{k-1}^{i}, & \mathrm{\mathrm{if}\,\,}\tilde{u}_{k}^{i}=0\,\,\mathrm{and}\,\,t_{k}^{i}<t_{k-1}^{i}+T^{i},\\
0,\quad & \mathrm{\mathrm{if}\,\,}\tilde{u}_{k}^{i}=0\,\,\mathrm{and}\,\,t_{k}^{i}\geq t_{k-1}^{i}+T^{i},
\end{cases}\label{eq:uhat-def}
\end{align}
 for $k\in\mathbb{N}$, where $T^{i}>0$. The scalar $\hat{u}_{k}^{i}$
denotes the control input to be applied after time $t_{k}^{i}$. 

To understand how the control input value is decided, we consider
a few cases. First, if the communication attempt at time $t_{k}^{i}$
is successful, then $\tilde{u}_{k}^{i}={\rm sign}_{\varepsilon}({\rm ave}^{i}(t_{k}^{i}))$,
otherwise $\tilde{u}_{k}^{i}=0$. If $\tilde{u}_{k}^{i}\neq0$ (or
equivalently if $\tilde{u}_{k}^{i}\in\{-1,1\}$), then the control
input value is set as $\hat{u}_{k}^{i}=\tilde{u}_{k}^{i}$. On the
other hand, if $\tilde{u}_{k}^{i}=0$ and the times $t_{k}^{i}$ and
$t_{k-1}^{i}$ are sufficiently close so that $t_{k}^{i}<t_{k-1}^{i}+T^{i}$,
then $\tilde{u}_{k-1}^{i}$ is applied as a control input (i.e., $\hat{u}_{k}^{i}=\tilde{u}_{k-1}^{i}$).
Otherwise, if $\tilde{u}_{k}^{i}=0$ and $t_{k}^{i}\geq t_{k-1}^{i}+T^{i}$,
then the control input is set to $\hat{u}_{k}^{i}=0$. 

The control input $\hat{u}_{k}^{i}$ is applied during the interval
$[t_{k}^{i},\hat{t}_{k}^{i})$, where $\hat{t}_{k}^{i}\triangleq\min\{t_{k+1}^{i},\tilde{t}_{k}^{i}\}$
with $\tilde{t}_{k}^{i}\triangleq t_{k}^{i}+T^{i}$ for $k=0$, and
\begin{align}
\tilde{t}_{k}^{i} & \triangleq\begin{cases}
t_{k}^{i}+T^{i}, & \mathrm{if}\,\,\tilde{u}_{k}^{i}\neq0,\\
t_{k-1}^{i}+T^{i}, & \mathrm{if}\,\,\tilde{u}_{k}^{i}=0\,\,\mathrm{and}\,\,t_{k}^{i}<t_{k-1}^{i}+T^{i},\\
t_{k}^{i}+T^{i},\quad & \mathrm{if}\,\,\tilde{u}_{k}^{i}=0\,\,\mathrm{and}\,\,t_{k}^{i}\geq t_{k-1}^{i}+T^{i},
\end{cases}\label{eq:ttilde-def}
\end{align}
 for $k\in\mathbb{N}$. Thus, the continuous-time control input $u^{i}(t)$
for each agent $i$ is given by $u^{i}(t)=0$ for $t\in[0,t_{0}^{i})$,
and 
\begin{align}
u^{i}(t) & \triangleq\begin{cases}
\hat{u}_{k}^{i},\quad & t\in[t_{k}^{i},\hat{t}_{k}^{i}),\\
0,\quad & t\in[\hat{t}_{k}^{i},t_{k+1}^{i}),
\end{cases}\quad k\in\mathbb{N}_{0}.\label{eq:control-input}
\end{align}
 As we will discuss in the following sections, the scalar $T^{i}>0$
and the communication attempt time instants $t_{0}^{i},t_{1}^{i},\ldots$
in our framework are selected to satisfy $t_{k}^{i}+T^{i}<t_{k+2}^{i}$
for every $k\in\mathbb{N}_{0}$. Notice that at time $t_{k}^{i}$,
the control input is set to either one of the values ${\rm sign}_{\varepsilon}({\rm ave}^{i}(t_{k}^{i}))$,
${\rm sign}_{\varepsilon}({\rm ave}^{i}(t_{k-1}^{i}))$, or $0$,
depending on the status of the communication attempts at $t_{k}^{i}$
and $t_{k-1}^{i}$. Moreover, this control input is applied until
time $\hat{t}_{k}^{i}=\min\{t_{k+1}^{i},\tilde{t}_{k}^{i}\}$. In
cases where $\tilde{t}_{k}^{i}<t_{k+1}^{i}$, the control input is
set to zero at $\tilde{t}_{k}^{i}$ until the next communication attempt
time $t_{k+1}^{i}$. Furthermore, in the cases where $t_{k+1}^{i}<\tilde{t}_{k}^{i}$,
agent $i$ first attempts a communication with its neighbors at time
$t_{k+1}^{i}$, and then updates the control input based on \eqref{eq:control-input}
with $k$ replaced by $k+1$. 

In our approach, communication attempt time instants are random variables
with certain distributions and their values are selected by the agents
using random number generators. In particular, for each agent $i$,
the first communication attempt time $t_{0}^{i}$ can be selected
at time $0$. Moreover, for each $k$, the communication attempt time
$t_{k+1}^{i}$ can be selected at time $t_{k}$.  Clearly, agent $i$
can check if $t=t_{k+1}^{i}$ at each time $t>t_{k}^{i}$. The control
input can thus be updated following the description in \eqref{eq:control-input}.
Notice that in certain cases, agent $i$ keeps the same control input
that is set at time $t_{k}^{i}$ even after time $t_{k+1}^{i}$. This
happens if the previous successful communication time $t_{k}^{i}$
is sufficiently close to $t_{k+1}^{i}$ and $\tilde{u}_{k+1}^{i}=0$
(i.e., $\varphi_{k+1}^{i}{\rm sign}_{\varepsilon}({\rm ave}^{i}(t_{k+1}^{i}))=0$).
In such cases, we have $u(t)=\hat{u}_{k+1}^{i}=\tilde{u}_{k}^{i}={\rm sign}_{\varepsilon}({\rm ave}^{i}(t_{k}^{i}))$
for $t\in[t_{k+1}^{i},\hat{t}_{k+1}^{i})$ with $\hat{t}_{k+1}^{i}=t_{k}^{i}+T^{i}$,
which indicates that the control input ${\rm sign}_{\varepsilon}({\rm ave}^{i}(t_{k}^{i}))$
is used until $t_{k}^{i}+T^{i}$ (which is after $t_{k+1}^{i}$). 

\begin{figure}
\begin{centering}
\includegraphics[width=0.9\columnwidth]{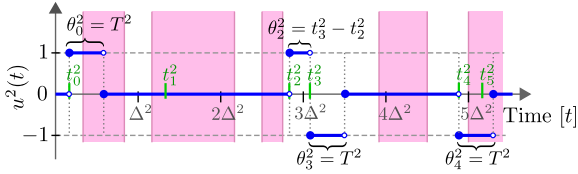} \vskip -9pt
\par\end{centering}
\caption{Control input trajectory of agent $2$ versus time. Jamming attack
intervals are denoted with shaded (pink) regions. }
\label{figure:control_input} \vskip -4pt
\end{figure}

\begin{exmp} In Figure~\ref{figure:control_input}, we show an example
trajectory for the control input $u^{2}(t)$ of agent $2$. In this
example, we have $\varphi_{0}^{2}=\varphi_{2}^{2}=\varphi_{3}^{2}=\varphi_{4}^{2}=1$
indicating that communication attempts at times $t_{0}^{2}$, $t_{2}^{2}$,
$t_{3}^{2}$, $t_{4}^{2}$ are not attacked and thus successful. For
this example scenario, at the successful communication time $t_{2}^{2}$,
agent $2$ sets its control input to $1$ based on the neighbor agents'
states. This control input is then applied until time $t_{3}^{2}$,
at which another successful communication is made and the control
input value is changed to $-1$. The new control input value is kept
constant until time $t_{3}^{2}+T^{2}$ and set to zero at that time. 

In our control approach, if agent $i$ makes a successful communication
at time $t_{k}^{i}$ and sets a nonzero control input, then this input
is kept constant at most until time $t_{k}^{i}+T^{i}$. If no new
successful communication is made until that time, the control input
is switched to zero at time $t_{k}^{i}+T^{i}$. If a successful communication
is made before (i.e., $t_{k+1}^{i}<t_{k}^{i}+T^{i}$) and $|{\rm ave}^{i}(t_{k+1}^{i})|\geq\varepsilon$,
then the control input value is set at time $t_{k+1}^{i}$ to $\tilde{u}_{k+1}^{i}={\rm sign}_{\varepsilon}({\rm ave}^{i}(t_{k+1}^{i}))$.
In the example scenario depicted in Figure~\ref{figure:control_input},
the nonzero control inputs that are set at instants $t_{0}^{2}$,
$t_{3}^{2}$, and $t_{4}^{2}$ are applied for $T^{2}$ units of time.
The control input is switched to zero at times $t_{0}^{2}+T^{2}$,
$t_{3}^{2}+T^{2}$, and $t_{4}^{2}+T^{2}$, since no new successful
communication attempts yielding new nonzero control inputs are made
during the intervals $[t_{0}^{2},t_{0}^{2}+T^{2})$, $[t_{3}^{2},t_{3}^{2}+T^{2})$,
and $[t_{4}^{2},t_{4}^{2}+T^{2})$.  \end{exmp}

To achieve consensus under the control law \eqref{eq:control-input},
it is important to design the times $t_{k}^{i}$, $k\in\mathbb{N}_{0}$,
$i\in\mathcal{V}$, at which the agents attempt to communicate with
each other.  In this paper, we take a stochastic approach and design
these times in Section~\ref{sec:Stochastic-Communication-Protoco}
as random variables. 

\begin{rem} In \citeasnoun{senejohnny_2015}, the communication attempt
times $t_{k}^{i}$, $k\in\mathbb{N}_{0}$, $i\in\mathcal{V}$, are
determined based on a \emph{deterministic} self-triggering approach.
There, the minimum interval between consecutive communication attempts
for each agent is given by $\Delta^{*}>0$. It is observed in \cite{senejohnny_2015}
that if the attacker is allowed to attack at a frequency larger than
the maximum frequency of communication attempts (given by $\frac{1}{\Delta^{*}}$),
then all communication may be blocked even if Assumption~\ref{Assumption:duration_sync}
is satisfied. This is because as $t_{k}^{i}$ are deterministic times,
an attacker that is knowledgeable on how $t_{k}^{i}$ are decided
may be able to generate a strategy to pinpoint $t_{k}^{i}$ with attacks
of very short durations and preserve energy for the rest of the time.
To avoid this problem, \citeasnoun{senejohnny_2015} considers the
additional assumption that there exist $\eta\geq0$ and $\sigma<\frac{1}{\Delta^{*}}$
such that 
\begin{equation}
\mathcal{I}(\tau,t)\leq\eta+\sigma(t-\tau),\label{eq:freq-condition}
\end{equation}
for all $\tau\geq0$ and $t\geq\tau$, where $\mathcal{I}(\tau,t)\in\mathbb{N}_{0}$
denotes the number of attack intervals $\mathcal{A}_{n}$ in the time
frame $[\tau,t]$. The scalar $\sigma>0$ in \eqref{eq:freq-condition}
represents an upper-bound on the attack frequency in the long run.
Note that since $\sigma<\frac{1}{\Delta^{*}}$, the assumption \eqref{eq:freq-condition}
guarantees that the attack frequency in large time frames is smaller
than the frequency of communication attempts. By utilizing $\rho$
from \eqref{eq:attack-duration-condition} and $\sigma$ from \eqref{eq:freq-condition},
the main result in \cite{senejohnny_2015} shows that consensus is
achieved if 
\begin{align}
\rho+\sigma\Delta^{*} & <1.\label{eq:condition-for-duration-and-freq}
\end{align}

In the following section, we propose a stochastic communication protocol,
where communication attempt times $t_{k}^{i}$, $k\in\mathbb{N}_{0}$,
$i\in\mathcal{V}$, are decided randomly. We show that in this case
even if \eqref{eq:freq-condition} and \eqref{eq:condition-for-duration-and-freq}
are not satisfied due to high frequency attacks, consensus can still
be achieved. \end{rem}

\section{Stochastic Communication Protocol for Consensus Under Jamming Attacks
\label{sec:Stochastic-Communication-Protoco}}

\subsection{Stochastic Communication Protocol}

We propose a communication protocol where each agent attempts to communicate
with its neighbors at random times that are unknown to the attacker
until the agents attempt communication at those times. 

\begin{defn}[Stochastic communication protocol] \label{Def:Communication}
For each agent $i\in\mathcal{V}$, let $\Delta^{i}>0$ be a fixed
scalar, and set $t_{k}^{i}$, $k\in\mathbb{N}_{0}$, to be independent
random variables such that $t_{k}^{i}$ has uniform distribution on
the interval $[k\Delta^{i},(k+1)\Delta^{i})$. \end{defn}

In this communication protocol, each agent $i$ attempts to make transmission
to its neighbors once in every $\Delta^{i}$ period. The communication
attempt time $t_{k}^{i}$ for the interval $[k\Delta^{i},(k+1)\Delta^{i})$
is selected randomly at time $k\Delta^{i}$ by agent $i$. Due to
uniform distribution of $t_{k}^{i}$, $k\in\mathbb{N}_{0}$, we have
$\mathbb{E}[t_{k+1}^{i}-t_{k}^{i}]=\Delta^{i}$, that is, the duration
between consecutive communication attempts are $\Delta^{i}$ \emph{in
expectation}. However, note that $t_{k+1}^{i}-t_{k}^{i}\neq\Delta^{i}$,
almost surely. We remark that the attacker is allowed to know how
$t_{k}^{i}$ is distributed, but not the value of $t_{k}^{i}$ until
the communication is attempted. In the example of Figure~\ref{figure:attack_pattern}
with $3$ agents, the attacker is able to block the first communication
attempts of agents $1$ and $3$. However, the first attempt of agent
$2$ is successful. Thus, for this example, $\varphi_{0}^{1}=\varphi_{0}^{3}=0$
and $\varphi_{0}^{2}=1$. 

\subsection{Finite-Time Consensus Analysis}

In the ternary control approach, the speed of change in each agent's
state is at most $1$. Therefore, agents need to apply the ternary
control input to their dynamics for sufficiently long durations to
get closer to their neighbors. Notice that the duration of control
input application is affected by the number of successful communication
attempts. In particular, it is likely that an agent $i$ would apply
control input to its dynamics for a longer duration in total, if that
agent makes many successful communications with its neighbors. 

As we characterize in the control law \eqref{eq:utilde-def}--\eqref{eq:control-input},
for each $k\in\mathbb{N}_{0}$ with $\varphi_{k}^{i}=1$ and $|{\rm ave}^{i}(t_{k}^{i})|\geq\varepsilon$,
the control input for agent $i$ is set to a nonzero value at time
$t_{k}^{i}$. As a preliminary discussion before our consensus analysis
in Theorem~\ref{Theorem:consensus} below, we first investigate the
duration for which this nonzero control input is applied to agent
$i$'s dynamics. To this end we define $\theta_{k}^{i}$, $k\in\mathbb{N}_{0}$,
by 
\begin{align}
\theta_{k}^{i} & \triangleq\begin{cases}
t_{k+1}^{i}-t_{k}^{i}, & \mathrm{if}\,\,t_{k+1}^{i}-t_{k}^{i}<T^{i}\,\,\mathrm{and\,}\,\tilde{u}_{k+1}^{i}\neq0,\\
T^{i}, & \mathrm{otherwise}.
\end{cases}\label{eq:thetadef}
\end{align}
It follows from \eqref{eq:utilde-def}--\eqref{eq:control-input}
that for each $k\in\mathbb{N}_{0}$ with $\varphi_{k}^{i}=1$ and
$|{\rm ave}^{i}(t_{k}^{i})|\geq\varepsilon$, the control input $\hat{u}_{k}^{i}=\tilde{u}_{k}^{i}={\rm sign}_{\varepsilon}({\rm ave}^{i}(t_{k}^{i}))$
is applied for $\theta_{k}^{i}$ units of time after $t_{k}^{i}$.
The applied control input is either $1$ or $-1$ depending on the
states of agent $i$ and its neighbors at time $t_{k}^{i}$. Notice
that the selected control input does not change during the interval
$[t_{k}^{i},t_{k}^{i}+\theta_{k}^{i})$. We also note that $\theta_{k}^{i}$
is used for consensus analysis and its knowledge is not required for
control purposes. 

The scalar $\theta_{k}^{i}$ is always upper-bounded by the constant
$T^{i}>0$ in \eqref{eq:thetadef}. In fact $\theta_{k}^{i}$ can
either be $t_{k+1}^{i}-t_{k}^{i}$ or $T^{i}$, and the first case
where $\theta_{k}^{i}=t_{k+1}^{i}-t_{k}^{i}$ happens only if the
next communication attempt time $t_{k+1}^{i}$ is close to $t_{k}^{i}$
so that $t_{k+1}^{i}-t_{k}^{i}<T^{i}$. To explain this case, we first
consider a successful communication time instant $t_{k}^{i}$ with
$|{\rm ave}^{i}(t_{k}^{i})|\geq\varepsilon$. Notice that $\varphi_{k}^{i}=1$
and ${\rm sign}_{\varepsilon}({\rm ave}^{i}(t_{k}^{i}))\neq0$, and
therefore, the control input is set to a nonzero value at $t_{k}^{i}$
(e.g., $t_{2}^{2}$ in Figure~\ref{figure:control_input}). The next
communication is attempted at time $t_{k+1}^{i}$. If the communication
attempt at $t_{k+1}^{i}$ is successful (i.e., $\varphi_{k+1}^{i}=1$)
and the neighboring agents' states are sufficiently far from agent
$i$'s state such that $|{\rm ave}^{i}(t_{k+1}^{i})|\geq\varepsilon$,
then $\tilde{u}_{k+1}^{i}\neq0$, and thus the control input is updated
to $\tilde{u}_{k+1}^{i}$ at time $t_{k+1}^{i}$ (e.g., $t_{3}^{2}$
in Figure~\ref{figure:control_input}). If the communication attempt
at time $t_{k+1}^{i}$ fails (i.e., $\varphi_{k+1}^{i}=0$) or $|{\rm ave}^{i}(t_{k+1}^{i})|<\varepsilon$,
then $\theta_{k}^{i}=T^{i}$ and the control input is unchanged until
time $t_{k}^{i}+T^{i}$. In the case where $t_{k+1}^{i}-t_{k}^{i}<T^{i}$,
if $\varphi_{k+1}^{i}=0$ or $|{\rm ave}^{i}(t_{k+1}^{i})|<\varepsilon$,
then we have $\hat{u}_{k+1}=\tilde{u}_{k}={\rm sign}_{\varepsilon}({\rm ave}^{i}(t_{k}^{i}))$
by \eqref{eq:utilde-def} and \eqref{eq:uhat-def}, and moreover,
it follows from \eqref{eq:ttilde-def} that $\hat{t}_{k+1}^{i}=\min\{t_{k+2}^{i},\tilde{t}_{k+1}^{i}\}$
with $\tilde{t}_{k+1}^{i}=t_{k}^{i}+T^{i}$. 

By using properties of $\theta_{k}^{i}$, we obtain the following
key result, which establishes intervals in which the control input
is guaranteed to be nonzero. 

\begin{lem}\label{KeyLemma} Suppose $T^{i}\in(0,\Delta^{i})$. Then,
for every $k\in\mathbb{N}_{0}$ such that $\varphi_{k}^{i}=1$ and
$|{\rm ave}^{i}(t_{k}^{i})|\geq\varepsilon$, we have 
\begin{align}
|u^{i}(t)| & =1,\quad t\in[t_{k}^{i},t_{k}^{i}+T^{i}).\label{eq:ueq1}
\end{align}
 \end{lem}

\begin{proof}By \eqref{eq:utilde-def}--\eqref{eq:control-input}
and \eqref{eq:thetadef}, we have 
\begin{align}
|u^{i}(t)| & =1,\quad t\in[t_{k}^{i},t_{k}^{i}+\theta_{k}^{i}).\label{eq:first1}
\end{align}
Notice that the value of $\theta_{k}^{i}$ can either be $T^{i}$
or $t_{k+1}^{i}-t_{k}^{i}$. If $\theta_{k}^{i}=T^{i}$, then \eqref{eq:first1}
implies \eqref{eq:ueq1}. 

On the other hand, if $\theta_{k}^{i}=t_{k+1}^{i}-t_{k}^{i}$, then
we have 
\begin{align}
|u^{i}(t)|=1,\quad & t\in[t_{k}^{i},t_{k+1}^{i}).\label{eq:second1}
\end{align}
Moreover, \eqref{eq:thetadef} indicates that the case $\theta_{k}^{i}=t_{k+1}^{i}-t_{k}^{i}$
happens when $t_{k+1}^{i}-t_{k}^{i}\leq T^{i}$ and $\tilde{u}_{k+1}^{i}\neq0$.
In this case, since $\tilde{u}_{k+1}^{i}\neq0$, it follows from \eqref{eq:uhat-def}--\eqref{eq:control-input}
and \eqref{eq:thetadef}  that
\begin{align}
|u^{i}(t)| & =1,\quad t\in[t_{k+1}^{i},t_{k+1}^{i}+\theta_{k+1}^{i}).\label{eq:second2}
\end{align}
Now, as a consequence of \eqref{eq:second1} and \eqref{eq:second2},
we get 
\begin{align}
|u^{i}(t)| & =1,\quad t\in[t_{k}^{i},t_{k+1}^{i}+\theta_{k+1}^{i}).\label{eq:resultbefore}
\end{align}
Furthermore, since $t_{k+1}^{i}\geq t_{k}^{i}$ and $t_{k+2}^{i}-t_{k}^{i}\geq\Delta^{i}$,
\begin{align*}
 & t_{k+1}^{i}+\theta_{k+1}^{i}\geq t_{k+1}^{i}+\min\{t_{k+2}^{i}-t_{k+1}^{i},T^{i}\}\\
 & \quad=\min\{t_{k+2}^{i},t_{k+1}^{i}+T^{i}\}\geq\min\{t_{k}^{i}+\Delta^{i},t_{k}^{i}+T^{i}\}.
\end{align*}
Using this inequality together with $T^{i}<\Delta^{i}$, we obtain
$t_{k+1}^{i}+\theta_{k+1}^{i}\geq t_{k}^{i}+T^{i}$. Hence, \eqref{eq:ueq1}
follows from \eqref{eq:resultbefore}. \end{proof}

We are ready to present the first main result of this paper. The theorem
below provides conditions under which the multi-agent system \eqref{eq:scalar-dynamics},
\eqref{eq:control-input} achieves consensus. In particular, they
are given in terms of $T^{i}$ of the control law, as well as $\Delta^{i}$
and $\varphi_{k}^{i}$, associated with the communication protocol. 

\begin{thm}\label{Theorem:consensus} Consider the multi-agent system
\eqref{eq:scalar-dynamics}, \eqref{eq:control-input} with $T^{i}\in(0,\min\{\frac{\varepsilon}{2d^{i}},\Delta^{i}\})$
where $\varepsilon>0$. Assume that under the stochastic communication
protocol, it holds
\begin{align}
\mathbb{P}\Bigl[\sum_{k=0}^{\infty}\varphi_{k}^{i}\geq M\Bigr]=1,\label{eq:phi_condition}
\end{align}
 for all $M\in\mathbb{N}_{0}$ and $i\in\mathcal{V}$. Then $x(t)$
converges in finite time to a vector $x^{*}\in\mathbb{R}^{n}$ belonging
to the set $\mathcal{D}_{\varepsilon}$ given by \eqref{eq:finalset},
almost surely. \end{thm} 

\begin{proof} First, for each $i\in\mathcal{V}$, let $\mathcal{K}^{i}\triangleq\{k\in\mathbb{N}_{0}\colon\varphi_{k}^{i}=1,|\mathrm{ave}^{i}(t_{k}^{i})|\geq\varepsilon\}$
and $\mathcal{T}^{i}\triangleq\bigcup_{k\in\mathcal{K}^{i}}[t_{k}^{i},t_{k}^{i}+\theta_{k}^{i})$.
Notice that for a given $k\in\mathcal{K}^{i}$, the time $t_{k}^{i}$
corresponds to a successful communication attempt time of agent $i$
for which $|\mathrm{ave}^{i}(t_{k}^{i})|\geq\varepsilon$. Thus, $\mathcal{T}^{i}$
represents the union of all nonzero control input application intervals
after successful communication attempts of agent $i$. 

Now, for $t\in\mathcal{T}^{i}$, let $\pi^{i}(t)\triangleq\max\{t_{k}^{i}\colon t_{k}^{i}\leq t,k\in\mathcal{K}^{i}\}$.
Here, $\pi^{i}(t)$ corresponds to the agent $i$'s last successful
communication attempt time $t_{k}^{i}$ before time $t$ such that
$|\mathrm{ave}^{i}(t_{k}^{i})|\geq\varepsilon$. Furthermore, let
\begin{align}
\mathcal{W}(t) & \triangleq\{i\in\mathcal{V}\colon t\in\mathcal{T}^{i}\},\quad t\geq0.\label{eq:Wdef}
\end{align}
Observe that $\mathcal{W}(t)$ corresponds to the set of agents that
have a nonzero control input $u_{i}(t)$ at time $t$. 

Now, let $V(x)\triangleq(1/2)x^{{\rm T}}Lx$, $x\in\mathbb{R}^{n}$.
By \eqref{eq:scalar-dynamics} and \eqref{eq:control-input}, 
\begin{align}
\dot{V}(x(t)) & =x^{\mathrm{T}}(t)Lu(t)\nonumber \\
 & =-\sum_{i\in\mathcal{W}(t)}\mathrm{ave}^{i}(t)\mathrm{sign}_{\varepsilon}(\mathrm{ave}^{i}(\pi^{i}(t))),\label{eq:lyapunov-derivative}
\end{align}
for $t\geq0$, where $u(t)\triangleq[\begin{array}{cccc}
u^{1}(t) & u^{2}(t) & \cdots & u^{n}(t)\end{array}]^{\mathrm{T}}$.

Our next goal is to show $\mathrm{ave}^{i}(t)\mathrm{sign}_{\varepsilon}(\mathrm{ave}^{i}(\pi^{i}(t)))=|\mathrm{ave}^{i}(t)|$,
$i\in\mathcal{W}(t)$. To this end, we need to show that $\mathrm{sign}_{\varepsilon}(\mathrm{ave}^{i}(\pi^{i}(t)))=\mathrm{sign}(\mathrm{ave}^{i}(t))$,
$i\in\mathcal{W}(t)$. To show this equivalence, we utilize an important
property of ternary control law. Specifically, in the interval $[\pi^{i}(t),t]$,
the state of an agent $j\in\mathcal{V}$ can change by at most $t-\pi^{i}$,
that is, $x^{j}(t)\in[x^{j}(\pi^{i}(t))-(t-\pi^{i}),x^{j}(\pi^{i}(t))+(t-\pi^{i})]$.
By using this property and noting that agent $i$ has $d^{i}$ number
of neighbors, we obtain that if $\mathrm{ave}^{i}(\pi^{i}(t))\geq\varepsilon$,
then 
\begin{align}
 & \sum_{j\in\mathcal{N}_{i}}(x^{j}(t)-x^{i}(t))\nonumber \\
 & \quad\geq\sum_{j\in\mathcal{N}_{i}}\big(x^{j}(\pi^{i}(t))-(t-\pi^{i})-(x^{i}(\pi^{i}(t))+(t-\pi^{i}))\big)\nonumber \\
 & \quad=\mathrm{ave}^{i}(\pi^{i}(t))-2d^{i}(t-\pi^{i}(t))\geq\varepsilon-2d^{i}(t-\pi^{i}(t)).\label{eq:res1}
\end{align}
Since $t\in\mathcal{T}^{i}$ for $i\in\mathcal{W}(t)$ and $\theta_{k}^{i}\leq T^{i}$
for $k\in\mathbb{N}_{0}$, we have $t-\pi^{i}(t)<T^{i}$. Furthermore,
since $T^{i}<\varepsilon/2d^{i}$, we have $t-\pi^{i}(t)<\varepsilon/2d^{i}$.
Hence, it follows from \eqref{eq:res1} that if $\mathrm{ave}^{i}(\pi^{i}(t))\geq\varepsilon$,
then 
\begin{align}
\sum_{j\in\mathcal{N}_{i}}(x^{j}(t)-x^{i}(t)) & >\varepsilon-2d^{i}\frac{\varepsilon}{2d^{i}}=0.\label{eq:res2}
\end{align}
 Similarly, we can show that if $\mathrm{ave}^{i}(\pi^{i}(t))\leq-\varepsilon$,
then
\begin{align}
\sum_{j\in\mathcal{N}_{i}}(x^{j}(t)-x^{i}(t)) & <-(\varepsilon-2d^{i}\frac{\varepsilon}{2d^{i}})=0.\label{eq:res3}
\end{align}
 Noting that $|\mathrm{ave}^{i}(\pi^{i}(t))|\geq\varepsilon$ for
$i\in\mathcal{W}(t)$, we obtain from \eqref{eq:res2} and \eqref{eq:res3}
that $\mathrm{sign}_{\varepsilon}(\mathrm{ave}^{i}(\pi^{i}(t)))=\mathrm{sign}(\mathrm{ave}^{i}(t))$,
$i\in\mathcal{W}(t)$. Consequently, we have $\mathrm{ave}^{i}(t)\mathrm{sign}_{\varepsilon}(\mathrm{ave}^{i}(\pi^{i}(t)))=|\mathrm{ave}^{i}(t)|$,
$i\in\mathcal{W}(t)$. 

It then follows from \eqref{eq:lyapunov-derivative} that 
\begin{align*}
\dot{V}(x(t)) & =-\sum_{i\in\mathcal{W}(t)}|\mathrm{ave}^{i}(t)|,\quad t\geq0.
\end{align*}
Since $|\mathrm{ave}^{i}(t)|\geq\varepsilon-2d^{i}(t-\pi^{i}(t))\geq\varepsilon-2d^{i}T^{i}$,
$i\in\mathcal{W}(t)$, we have 
\begin{align}
\dot{V}(x(t)) & \leq-\sum_{i\in\mathcal{W}(t)}(\varepsilon-2d^{i}T^{i})\leq-\alpha\sum_{i\in\mathcal{W}(t)}1\nonumber \\
 & =-\alpha\sum_{i\in\mathcal{V}}g^{i}(t),\quad t\geq0,\label{eq:lyapunov-derivative-final-form}
\end{align}
 where $\alpha\triangleq\min_{i\in\mathcal{V}}(\varepsilon-2d^{i}T^{i})$
and $g^{i}(t)\triangleq\mathbbm{1}\big[t\in\mathcal{T}^{i}\big]$.
By integrating both sides of \eqref{eq:lyapunov-derivative-final-form},
\begin{align}
V(x(t)) & \leq V(x(0))-\alpha\int_{0}^{t}\sum_{i\in\mathcal{V}}g^{i}(s)\mathrm{d}s,\quad t\geq0.\label{eq:integral-ineq}
\end{align}
Since $V(x(t))\geq0$, it follows from \eqref{eq:integral-ineq} that
\begin{align}
\int_{0}^{t}\sum_{i\in\mathcal{V}}g^{i}(s)\mathrm{d}s & \leq\frac{V(x(0))}{\alpha},\quad t\geq0.\label{eq:epsilon-integral}
\end{align}

Now, let $H^{i}\triangleq\{\omega\in\Omega\colon\sum_{k=0}^{\infty}\varphi_{k}^{i}=\infty\}$
and $H\triangleq\cap_{i\in\mathcal{V}}H^{i}$. By \eqref{eq:phi_condition},
we have $\mathbb{P}[H^{i}]=1$, $i\in\mathcal{V}$, and thus, $\mathbb{P}[H]=1$.
In what follows, we show that finite time approximate consensus is
achieved for every $\omega\in H$. Define
\begin{align*}
\underline{k}^{i}(t) & \triangleq\inf\{k\in\mathbb{N}_{0}\colon t_{k}^{i}>t,\,\varphi_{k}^{i}=1\},
\end{align*}
 and $\underline{t}^{i}(t)\triangleq t_{\underline{k}^{i}(t)}$, $i\in\mathcal{V}$.
Here, $\underline{t}^{i}(t)$ denotes the first successful communication
time instant of agent $i$ after time $t$. Notice that for every
$\omega\in H$, we have $\underline{t}^{i}(t)<\infty$, $t\geq0$. 

Let $\underline{T}\triangleq\min_{i\in\mathcal{V}}T^{i}$ and $\overline{T}\triangleq\max_{i\in\mathcal{V}}T^{i}$.
Moreover, let $\vartheta_{0}\triangleq\overline{T}$, and 
\begin{align*}
\vartheta_{k+1} & \triangleq\max_{i\in\mathcal{V}}\underline{t}^{i}(\vartheta_{k})+2\overline{T},\quad k\in\mathbb{N}_{0}.
\end{align*}
First, we show that for every $\omega\in H$,
\begin{align}
\int_{\vartheta_{k}-\overline{T}}^{\vartheta_{k+1}-\overline{T}}\sum_{i\in\mathcal{V}}g^{i}(s)\mathrm{d}s & \geq\big(\max_{i\in\mathcal{V}}\mathbbm{1}[|\mathrm{ave}^{i}(\vartheta_{k})|\geq\varepsilon]\big)\underline{T}.\label{eq:glowerbound}
\end{align}
 Notice that $\int_{\vartheta_{k}-\overline{T}}^{\vartheta_{k+1}-\overline{T}}\sum_{i\in\mathcal{V}}g^{i}(s)\mathrm{d}s\geq0$,
since $g^{i}(s)\geq0$ for $i\in\mathcal{V}$. Therefore, in the case
where $\max_{i\in\mathcal{V}}\mathbbm{1}[|\mathrm{ave}^{i}(\vartheta_{k})|\geq\varepsilon]=0$,
we have \eqref{eq:glowerbound}. Now, consider the case where $\max_{i\in\mathcal{V}}\mathbbm{1}[|\mathrm{ave}^{i}(\vartheta_{k})|\geq\varepsilon]=1$.
In this case, there exists an agent $i^{*}\in\mathcal{V}$ such that
$|\mathrm{ave}^{i^{*}}(\vartheta_{k})|\geq\varepsilon$. This agent
makes a successful communication with its neighbors at time $\underline{t}^{i^{*}}(\vartheta_{k})\in(\vartheta_{k},\vartheta_{k+1})$.
If $|\mathrm{ave}^{i^{*}}(\underline{t}^{i^{*}}(\vartheta_{k}))|\geq\varepsilon$,
then since $T^{i}<\Delta^{i}$, by Lemma~\ref{KeyLemma}, we have
$|u^{i^{*}}(t)|=1$, $t\in[\underline{t}^{i^{*}}(\vartheta_{k}),\underline{t}^{i^{*}}(\vartheta_{k})+T^{i^{*}})$.
Since $|u^{i^{*}}(t)|=1$ implies $g^{i^{*}}(t)=1$, we have $g^{i^{*}}(t)=1$,
$t\in[\underline{t}^{i^{*}}(\vartheta_{k}),\underline{t}^{i^{*}}(\vartheta_{k})+T^{i^{*}})$.
As a result, 
\begin{align}
 & \int_{\vartheta_{k}-\overline{T}}^{\vartheta_{k+1}-\overline{T}}\sum_{i\in\mathcal{V}}g^{i}(s)\mathrm{d}s\geq\int_{\underline{t}^{i^{*}}(\vartheta_{k})}^{\underline{t}^{i^{*}}(\vartheta_{k})+T^{i}}\sum_{i\in\mathcal{V}}g^{i}(s)\mathrm{d}s\nonumber \\
 & \quad\geq\int_{\underline{t}^{i^{*}}(\vartheta_{k})}^{\underline{t}^{i^{*}}(\vartheta_{k})+T^{i}}g^{i^{*}}(s)\mathrm{d}s=T^{i^{*}}\geq\underline{T}.\label{eq:Tiineq}
\end{align}

On the other hand, if $|\mathrm{ave}^{i^{*}}(\underline{t}^{i^{*}}(\vartheta_{k}))|<\varepsilon$,
then there may be two cases: 1) agent $i^{*}$ was in the process
of changing its state at time $\vartheta_{k}$ and 2) the state of
another agent $j^{*}\in\mathcal{N}_{i^{*}}$ became closer to the
state of agent $i^{*}$ in the interval between the times $\vartheta_{k}$
and $\underline{t}^{i^{*}}(\vartheta_{k})$, since $|\mathrm{ave}^{i^{*}}(\vartheta_{k})|\geq\varepsilon$.
Case 1 implies that for some $\hat{k}\in\mathbb{N}_{0}$, we have
$\varphi_{\hat{k}}^{i^{*}}=1$ and $|\mathrm{ave}^{i^{*}}(t_{\hat{k}}^{i^{*}})|\geq\varepsilon$
with $t_{\hat{k}}^{i^{*}}\in[\vartheta_{k}-\overline{T},\vartheta_{k}]$.
Thus, by Lemma~\ref{KeyLemma}, we get $|u^{i^{*}}(t)|=1$ for $t\in[t_{\hat{k}}^{i^{*}},t_{\hat{k}}^{i^{*}}+T^{i^{*}})$
implying $g^{i^{*}}(t)=1$, $t\in[t_{\hat{k}}^{i^{*}},t_{\hat{k}}^{i^{*}}+T^{i^{*}})$.
Consequently, 
\begin{align}
 & \int_{\vartheta_{k}-\overline{T}}^{\vartheta_{k+1}-\overline{T}}\sum_{i\in\mathcal{V}}g^{i}(s)\mathrm{d}s\geq\int_{t_{\hat{k}}^{i^{*}}}^{t_{\hat{k}}^{i^{*}}+T^{i^{*}}}g^{i^{*}}(s)\mathrm{d}s=T^{i^{*}}\geq\underline{T}.\label{eq:Tiineq2}
\end{align}
Case 2 implies that agent $j^{*}$ updated its state at least once
in the interval $[\vartheta_{k}-\overline{T},\underline{t}^{i^{*}}(\vartheta_{k})+T^{j^{*}})$.
In other words, for some $\tilde{k}\in\mathbb{N}_{0}$, we have $\varphi_{\tilde{k}}^{j^{*}}=1$
and $|\mathrm{ave}^{j^{*}}(t_{\tilde{k}}^{j^{*}})|\geq\varepsilon$
with $t_{\tilde{k}}^{j^{*}}\in[\vartheta_{k}-\overline{T},\underline{t}^{i^{*}}(\vartheta_{k})+T^{j^{*}})$.
It then follows by Lemma~\ref{KeyLemma} that $|u^{j^{*}}(t)|=1$
for $t\in[t_{\tilde{k}}^{j^{*}},t_{\tilde{k}}^{j^{*}}+T^{j^{*}})$.
Therefore,  $g^{j^{*}}(t)=1$ for $t\in[t_{\tilde{k}}^{j^{*}},t_{\tilde{k}}^{j^{*}}+T^{j^{*}})$,
and hence,
\begin{align}
\int_{\vartheta_{k}-\overline{T}}^{\vartheta_{k+1}-\overline{T}}\sum_{i\in\mathcal{V}}g^{i}(s)\mathrm{d}s & \geq\int_{t_{\tilde{k}}^{j^{*}}}^{t_{\tilde{k}}^{j^{*}}+T^{j^{*}}}g^{j^{*}}(s)\mathrm{d}s=T^{j^{*}}\geq\underline{T}.\label{eq:Tjineq}
\end{align}
The inequalities \eqref{eq:Tiineq}--\eqref{eq:Tjineq} that are
obtained for different cases imply \eqref{eq:glowerbound}. 

It follows from \eqref{eq:epsilon-integral} and \eqref{eq:glowerbound}
that for every $\omega\in H$, 
\begin{align*}
 & \sum_{k=0}^{\infty}\max_{i\in\mathcal{V}}\mathbbm{1}[|\mathrm{ave}^{i}(\vartheta_{k})|\geq\varepsilon]\leq\frac{1}{\underline{T}}\sum_{k=0}^{\infty}\int_{\vartheta_{k}-\overline{T}}^{\vartheta_{k+1}-\overline{T}}\sum_{i\in\mathcal{V}}g^{i}(s)\mathrm{d}s\\
 & \quad=\int_{0}^{\infty}\sum_{i\in\mathcal{V}}g^{i}(s)\mathrm{d}s\leq\frac{1}{\underline{T}}\frac{V(x(0))}{\alpha},
\end{align*}
 which implies that there exists $k^{*}\in\mathbb{N}_{0}$ such that
for every $k\in\{k^{*},k^{*}+1,\ldots\}$, $\max_{i\in\mathcal{V}}\mathbbm{1}[|\mathrm{ave}^{i}(\vartheta_{k})|\geq\varepsilon]=0.$
Thus, for each agent $i\in\mathcal{V}$, we have $|\mathrm{ave}^{i}(\vartheta_{k})|<\varepsilon$
and $|\mathrm{ave}^{i}(\vartheta_{k+1})|<\varepsilon$ for $k\in\{k^{*},k^{*}+1,\ldots\}$.
Therefore, $\max_{t\in[\vartheta_{k},\vartheta_{k+1})}\max_{i\in\mathcal{V}}\mathbbm{1}[|\mathrm{ave}^{i}(t)|\geq\varepsilon]=1$
implies that at least one agent $i^{*}\in\mathcal{V}$ started changing
its state at some time $\tilde{t}_{k}^{i^{*}}\in(\vartheta_{k},\vartheta_{k+1})$,
and thus by Lemma~\ref{KeyLemma}, we have $|u^{i^{*}}(t)|=1$ for
$t\in[\tilde{t}_{k}^{i^{*}},\tilde{t}_{k}^{i^{*}}+T^{i^{*}})$ implying
$g^{i^{*}}(t)=1$, $t\in[\tilde{t}_{k}^{i^{*}},\tilde{t}_{k}^{i^{*}}+T^{i^{*}})$.
Consequently, by using an argument similar to the one that we used
for obtaining \eqref{eq:glowerbound}, we obtain
\begin{align*}
 & \max_{t\in[\vartheta_{k},\vartheta_{k+1})}\max_{i\in\mathcal{V}}\mathbbm{1}[|\mathrm{ave}^{i}(t)|\geq\varepsilon]\\
 & \quad\leq\frac{1}{\underline{T}}\int_{\vartheta_{k}}^{\vartheta_{k+1}+\overline{T}}\sum_{i\in\mathcal{V}}g^{i}(s)\mathrm{d}s\leq\frac{1}{\underline{T}}\int_{\vartheta_{k}-\overline{T}}^{\vartheta_{k+2}-\overline{T}}\sum_{i\in\mathcal{V}}g^{i}(s)\mathrm{d}s\\
 & \quad=\frac{1}{\underline{T}}\int_{\vartheta_{k}-\overline{T}}^{\vartheta_{k+1}-\overline{T}}\sum_{i\in\mathcal{V}}g^{i}(s)\mathrm{d}s+\frac{1}{\underline{T}}\int_{\vartheta_{k+1}-\overline{T}}^{\vartheta_{k+2}-\overline{T}}\sum_{i\in\mathcal{V}}g^{i}(s)\mathrm{d}s,
\end{align*}
for every $k\in\{k^{*},k^{*}+1,\ldots\}$. Now, by using this inequality
together with \eqref{eq:epsilon-integral}, 
\begin{align*}
 & \sum_{k=k^{*}}^{\infty}\max_{t\in[\vartheta_{k},\vartheta_{k+1})}\max_{i\in\mathcal{V}}\mathbbm{1}[|\mathrm{ave}^{i}(t)|\geq\varepsilon]\\
 & \quad\leq\frac{1}{\underline{T}}\sum_{k=k^{*}}^{\infty}\int_{\vartheta_{k}-\overline{T}}^{\vartheta_{k+1}-\overline{T}}\sum_{i\in\mathcal{V}}g^{i}(s)\mathrm{d}s\\
 & \quad\quad+\frac{1}{\underline{T}}\sum_{k=k^{*}+1}^{\infty}\int_{\vartheta_{k}-\overline{T}}^{\vartheta_{k+1}-\overline{T}}\sum_{i\in\mathcal{V}}g^{i}(s)\mathrm{d}s\\
 & \quad\leq\frac{2}{\underline{T}}\int_{0}^{\infty}\sum_{i\in\mathcal{V}}g^{i}(s)\mathrm{d}s\leq\frac{2}{\underline{T}}\frac{V(x(0))}{\alpha},
\end{align*}
which implies that for every $\omega\in H$, there exists $\tilde{k}^{*}\in\{k^{*},k^{*}+1,\ldots\}$
such that 
\begin{align*}
\max_{t\in[\vartheta_{k},\vartheta_{k+1})}\max_{i\in\mathcal{V}}\mathbbm{1}[|\mathrm{ave}^{i}(t)| & \geq\varepsilon]=0,\,\,k\in\{\tilde{k}^{*},\tilde{k}^{*}+1,\ldots\}.
\end{align*}
 Thus, for every $\omega\in H$, we have $|\mathrm{ave}^{i}(t)|<\varepsilon$
for $i\in\mathcal{V}$ and $t\geq\vartheta_{\tilde{k}^{*}}$, implying
almost sure finite time approximate consensus, since $\mathbb{P}[H]=1$.
\end{proof} 

Condition \eqref{eq:phi_condition} in Theorem~\ref{Theorem:consensus}
is concerned with the long run statistics of inter-agent communication
attempts. Specifically, $\sum_{k=0}^{\infty}\varphi_{k}^{i}$ in \eqref{eq:phi_condition}
represents the total number of successful communication attempts of
agent $i$. Under the condition \eqref{eq:phi_condition}, this number
exceeds all nonnegative integers almost surely, indicating that each
agent $i$ can achieve infinitely many successful communications with
its neighbors in the long run. In Section~\ref{sec:Deterministic-Jamming-and},
we will show that \eqref{eq:phi_condition} holds and consensus can
be achieved under different attack strategies that satisfy Assumption~\ref{Assumption:duration_sync}
on the average duration of jamming.

For the proof of Theorem~\ref{Theorem:consensus}, we utilize the
function $(1/2)x^{\mathrm{T}}(t)Lx(t)$ and explore its evolution.
This approach is also utilized by \citeasnoun{de_persis_ternary_2013},
\citeasnoun{senejohnny_2015}, and \citeasnoun{senejohnny2017jamming}.
However, there are some key differences. An important role is played
by Lemma~\ref{KeyLemma}. Notice also that the analysis is facilitated
by the choice of parameters $T^{i}\in(0,\min\{\frac{\varepsilon}{2d^{i}},\Delta^{i}\})$,
$i\in\mathcal{V}$. In particular, $T^{i}<\Delta^{i}$ allows us to
utilize Lemma~\ref{KeyLemma} to characterize the durations for which
agent $i$ has a nonzero control input. Furthermore, $T^{i}<\frac{\varepsilon}{2d^{i}}$
ensures a certain decay for the Lyapunov-like function $(1/2)x^{\mathrm{T}}(t)Lx(t)$
after a successful communication by agent $i$. 

In our ternary control approach, after each successful transmission
at $t_{k}^{i}$, agent $i$ may only apply a constant and bounded
control input for a maximum duration of $T^{i}$. Thus the worst-case
change in relative state positions ($\sum_{j\in\mathcal{N}_{i}}(x^{j}(t)-x^{i}(t))$)
during a control input application by agent $i$ can be calculated,
since all agents can only change their states with speed $1$. This
worst case is taken into account in the design of $T^{i}$ for agent
$i$. If instead of the ternary input, real-valued ${\rm ave}^{i}(t_{k}^{i})$
is used, then agent $i$ would not be able to know how its neighbors
may move. This is because the control input of a neighboring agent
$j\in\mathcal{N}^{i}$ depends on other agents' states that agent
$i$ does not have access to. As a result, agent $i$ may not be able
to choose an appropriate control input application duration to guarantee
a decrease in the Lyapunov-like function. 

We note that the notion of \emph{finite-time approximate consensus}
that we utilize in Theorem~\ref{Theorem:consensus} is similar to
the notion of \emph{finite-time contractive stability} discussed in
\citeasnoun{dorato2006overview} for nonlinear dynamical systems.
In particular, agent states that are initially outside the approximate
consensus set $\mathcal{D}_{\varepsilon}$ enter $\mathcal{D}_{\varepsilon}$
in finite time and stay there, with probability one. This approximate
consensus time is a random variable that depends on the initial agent
states, the network topology, the communication attempt times, as
well as the jamming times and durations. 

Our control approach and hence the consensus time achieved with it
have different characteristics from the ones in \citeasnoun{de_persis_ternary_2013},
\citeasnoun{senejohnny_2015}, and \citeasnoun{senejohnny2017jamming}.
In those works, the control input values are changed at communication
attempt time instants and at each communication attempt time $t_{k}^{i}$,
agent $i$ decides the next communication time based on the average
distance $|\mathrm{ave}^{i}(t_{k})|$ from its neighbors. If this
distance is large, then the duration until the next communication
time becomes large. On the other hand, in our approach, when agent
$i$ makes a successful communication at time $t_{k}^{i}$ and sets
a nonzero control input, this input is not necessarily kept constant
until the next communication time $t_{k+1}^{i}$. As we discuss above,
the time $t_{k+1}^{i}$ may be far from $t_{k}^{i}$. In such cases,
the control input is set to zero at time $t_{k}^{i}+T^{i}$ and zero
input is used until time $t_{k+1}^{i}$. With this approach, we are
allowed to pick large $\Delta^{i}$ values in our randomized communication
approach to increase the expected interval length between consecutive
communication attempts and reduce the number of transmissions. Consensus
is guaranteed as long as $T^{i}\in(0,\min\{\frac{\varepsilon}{2d^{i}},\Delta^{i}\}),$
$i\in\mathcal{V}$. In the absence of jamming attacks, the approximate
consensus time achieved with our control approach can be larger than
the time achieved with the approach of \citeasnoun{de_persis_ternary_2013},
since in our approach, control inputs of agents may be set to zero
between consecutive communication times. We note that to reduce the
expected total duration for which the control input is zero, $T^{i}$
can be selected close to $\Delta^{i}$, which also reduces the time
when consensus is achieved. 

\section{Deterministic Jamming and Communication-Aware Jamming \label{sec:Deterministic-Jamming-and}}

In this section, we consider two different attack strategies that
a jamming attacker may follow. We show that consensus can be achieved
in both cases. 

\subsection{Consensus Under Deterministic Attacks \label{subsec:Consensus-Under-Deterministic}}

First, we consider the attack strategy where the starting time and
the duration of the jamming attacks do not depend on the time instants
at which the agents try to communicate. In particular, concerning
the sequences $\{a_{k}\}_{k\in\mathbb{N}_{0}}$ and $\{\tau_{k}\}_{k\in\mathbb{N}_{0}}$,
we assume the following. 

\begin{assum} \label{assumption:deterministic_sync} The sequences
$\{a_{k}\}_{k\in\mathbb{N}_{0}}$ and $\{\tau_{k}\}_{k\in\mathbb{N}_{0}}$,
which characterize the jamming attacks, are decided deterministically,
that is, for every $\omega\in\Omega$ and $k\in\mathbb{N}_{0}$, 
\begin{align}
a_{k}(\omega) & =\bar{a}_{k},\quad\tau_{k}(\omega)=\bar{\tau}_{k},
\end{align}
where $\bar{a}_{k}\geq0$ and $\bar{\tau}_{k}\geq0$ for $k\in\mathbb{N}_{0}$
are fixed scalars. \end{assum}

Assumption~\ref{assumption:deterministic_sync} is useful to model
scenarios where the attacker cannot detect the transmissions on the
communication channels. Note that the attacker may still be knowledgeable
on certain properties of the multi-agent system such as the number
of agents, communication topology, as well as the scalars $\Delta^{i}$,
$i\in\mathcal{V}$, used in the communication protocol. 

Our analysis relies on a few key definitions. First, let 
\begin{align}
\gamma^{i} & \triangleq\min\Big\{ k\in\mathbb{N}\colon k\Delta^{i}>\kappa/(1-\rho)\Big\},\quad i\in\mathcal{V}.\label{eq:gamma-def}
\end{align}
Now define $\hat{\Delta}^{i}>0$ and $\hat{\varphi}_{k}^{i}\in\{0,1\}$,
$k\in\mathbb{N}_{0}$, by 
\begin{align}
\hat{\Delta}^{i} & \triangleq\gamma^{i}\Delta^{i},\label{eq:Delta-hat-def}\\
\hat{\varphi}_{k}^{i} & \triangleq\begin{cases}
0, & \mathrm{if}\,\,\varphi_{k\gamma^{i}}^{i}=0,\cdots,\varphi_{(k+1)\gamma^{i}-1}^{i}=0,\\
1, & \mathrm{otherwise}.
\end{cases}\label{eq:phi_hat_def}
\end{align}
 With these definitions, $\hat{\Delta}^{i}$ is an integer multiple
of $\Delta^{i}$ that is selected to be larger than $\kappa/(1-\rho)$.
In the interval $[k\hat{\Delta}^{i},(k+1)\hat{\Delta}^{i})$, agent
$i$ makes $\gamma^{i}$ number of communication attempts with its
neighbors, and moreover, $\hat{\varphi}_{k}^{i}$ takes the value
$0$ if all of these attempts fail and $1$ if one or more of these
attempts are successful. We emphasize that $\gamma^{i}$, $\hat{\Delta}^{i}$,
and $\hat{\varphi}_{k}^{i}$ are used only for the purpose of analysis,
and their values are not needed in our stochastic communication protocol. 

We now show that under Assumptions \ref{Assumption:duration_sync}
and \ref{assumption:deterministic_sync}, agents can successfully
communicate with their neighbors infinitely many times in the long
run, almost surely.

\begin{prop} \label{proposition:persistency_async} For any jamming
attacks described by sequences $\{a_{k}\}_{k\in\mathbb{N}_{0}}$ and
$\{\tau_{k}\}_{k\in\mathbb{N}_{0}}$ that satisfy Assumptions \ref{Assumption:duration_sync}
and \ref{assumption:deterministic_sync}, the equality in \eqref{eq:phi_condition}
holds. \end{prop}

\begin{proof} It follows from \eqref{eq:phi_hat_def} that for every
$i\in\mathcal{V}$, 
\begin{align}
\mathbb{P}\Bigl[\sum_{k=0}^{\infty}\varphi_{k}^{i}\geq M\Bigr] & \geq\mathbb{P}\Bigl[\sum_{k=0}^{\infty}\hat{\varphi}_{k}^{i}\geq M\Bigr],\quad M\in\mathbb{N}_{0}.\label{eq:p-ineq}
\end{align}
 In what follows, we show \eqref{eq:phi_condition} by proving that
\begin{equation}
\mathbb{P}\Bigl[\sum_{k=0}^{\infty}\hat{\varphi}_{k}^{i}\geq M\Bigr]=1,\quad M\in\mathbb{N}_{0},\quad i\in\mathcal{V}.\label{eq:lemma-result}
\end{equation}

First, let $B_{k}^{i}\triangleq\{\omega\in\Omega:\hat{\varphi}_{k}^{i}(\omega)=1\}$,
$k\in\mathbb{N}_{0}$, and $E\triangleq\bigcap_{l=0}^{\infty}\Bigl(\bigcup_{k\geq l}B_{k}^{i}\Bigr)$.
Furthermore, for each $k\in\mathbb{N}_{0}$, let $\beta_{k}^{i}\colon\Omega\to\{0,1,\ldots,\gamma^{i}-1\}$
be a random variable distributed according to $\mathbb{P}[\beta_{k}^{i}=l]=1/\gamma^{i}$
for each $l\in\{0,1,\ldots,\gamma^{i}-1\}$, and define $\hat{t}_{k}^{i}\triangleq t_{k\gamma^{i}+\beta_{k}^{i}}^{i}$,
$k\in\mathbb{N}_{0}$. 

Note that $\hat{t}_{k}^{i}$ is a random variable distributed uniformly
in $[k\hat{\Delta}^{i},(k+1)\hat{\Delta}^{i})$. Since $B_{k}^{i}=\bigcup_{l=0}^{\gamma^{i}-1}\{t_{k\gamma^{i}+l}^{i}\in\overline{\mathcal{A}}^{\mathrm{c}}(k\hat{\Delta}^{i},(k+1)\hat{\Delta}^{i})\}$,
we have $B_{k}^{i}\supseteq\{\hat{t}_{k}^{i}\in\overline{\mathcal{A}}^{\mathrm{c}}(k\hat{\Delta}^{i},(k+1)\hat{\Delta}^{i})\}$.
Hence,
\begin{align*}
\mathbb{P}[B_{k}^{i}] & \geq\mathbb{P}[\hat{t}_{k}^{i}\in\overline{\mathcal{A}}^{\mathrm{c}}(k\hat{\Delta}^{i},(k+1)\hat{\Delta}^{i})]\\
 & =\mathbb{P}[\hat{t}_{k}^{i}\notin\overline{\mathcal{A}}(k\hat{\Delta}^{i},(k+1)\hat{\Delta}^{i})]\\
 & =\frac{\hat{\Delta}^{i}-|\overline{\mathcal{A}}(k\hat{\Delta}^{i},(k+1)\hat{\Delta}^{i})|}{\hat{\Delta}^{i}},\quad k\in\mathbb{N}_{0}.
\end{align*}
 By Assumption \ref{Assumption:duration_sync} and $\hat{\Delta}^{i}>\kappa/(1-\rho)$,
we obtain 
\begin{align*}
\mathbb{P}[B_{k}^{i}] & \geq\frac{\hat{\Delta}^{i}-\kappa-\rho\hat{\Delta}^{i}}{\hat{\Delta}^{i}}=1-\rho-\cfrac{\kappa}{\hat{\Delta}^{i}}>0,\quad k\in\mathbb{N}_{0}.
\end{align*}
As a consequence, $\sum_{k=0}^{\infty}\mathbb{P}[B_{k}^{i}]=\infty$.
Now, since $t_{0}^{i},t_{1}^{i},\ldots$ are independent and the sequences
$\{a_{k}\}_{k\in\mathbb{N}_{0}}$ and $\{\tau_{k}\}_{k\in\mathbb{N}_{0}}$
are deterministic (by Assumption~\ref{assumption:deterministic_sync}),
the events $B_{0}^{i},B_{1}^{i},\ldots$ are independent. Therefore,
it follows from $\sum_{k=0}^{\infty}\mathbb{P}[B_{k}^{i}]=\infty$
and the Borel-Cantelli Lemma (see Theorem 3.22 of \citeasnoun{karrprobabilitybook})
that $\mathbb{P}[E]=1$. Consequently, noting that $\{\omega\in\Omega\colon\sum_{k=0}^{\infty}\hat{\varphi}_{k}^{i}(\omega)\geq M\}\supseteq E$,
we obtain $\mathbb{P}[\sum_{k=0}^{\infty}\hat{\varphi}_{k}^{i}\geq M]\geq\mathbb{P}[E]$.
Hence, \eqref{eq:lemma-result} holds. Finally, \eqref{eq:phi_condition}
follows from \eqref{eq:p-ineq} and \eqref{eq:lemma-result}. \end{proof}

Proposition~\ref{proposition:persistency_async} implies that agents
can achieve infinitely many successful communications with their neighbors
in the long run under any deterministic attack strategy satisfying
\eqref{eq:attack-duration-condition} in Assumption \ref{Assumption:duration_sync}. 

The proof of Proposition~\ref{proposition:persistency_async} relies
on a few essential principles. First of all, we do not directly compute
the successful communications in each $\Delta^{i}$-length interval.
Instead, we look at the longer $\hat{\Delta}^{i}$-length intervals
and compute how many of these intervals include successful communications.
This is useful due to the fact that $\hat{\Delta}^{i}$ is chosen
for the analysis to be larger than the longest possible duration $\kappa/(1-\rho)$
of a continuous jamming attack. Regardless of how large $\kappa\geq0$
and $\rho\in(0,1)$ can be, there always exists such a $\hat{\Delta}^{i}$
as given in \eqref{eq:Delta-hat-def}. We remark again that since
$\hat{\Delta}^{i}$ is needed only for the analysis, its value is
not necessary for the multi-agent operation. 

In the proof of Proposition~\ref{proposition:persistency_async},
we also take advantage of the uniform distribution of the communication
attempt times in each $\Delta^{i}$-length interval. The probability
of at least one successful communication in a $\hat{\Delta}^{i}$-length
interval $[k\hat{\Delta}^{i},(k+1)\hat{\Delta}^{i})$ is lower-bounded
by the probability of an event that we construct in the proof. This
is the event that one of communication attempt times that is selected
uniformly randomly from the $\gamma^{i}$ number of attempt times
in the interval $[k\hat{\Delta}^{i},(k+1)\hat{\Delta}^{i})$ does
not face a jamming attack. The uniform distribution property of the
attempt times over the $\Delta^{i}$-length intervals and thus the
$\hat{\Delta}^{i}$-length intervals allows derivation of the probability
bound. 

The proof of Proposition~\ref{proposition:persistency_async} also
relies on the fact that the attacks are deterministic, and hence the
attack times do not depend on the communication attempt times. 

Next, by using Proposition~\ref{proposition:persistency_async} and
Theorem~\ref{Theorem:consensus}, we show that under deterministic
attacks, the multi-agent system \eqref{eq:scalar-dynamics}, \eqref{eq:control-input}
achieves consensus in finite time, almost surely. 

\begin{thm}\label{Theorem:consensus-deterministic} Consider the
multi-agent system \eqref{eq:scalar-dynamics}, \eqref{eq:control-input}
with $T^{i}\in(0,\min\{\frac{\varepsilon}{2d^{i}},\Delta^{i}\})$
where $\varepsilon>0$. For any jamming attacks described by sequences
$\{a_{k}\}_{k\in\mathbb{N}_{0}}$ and $\{\tau_{k}\}_{k\in\mathbb{N}_{0}}$
that satisfy Assumptions \ref{Assumption:duration_sync} and \ref{assumption:deterministic_sync},
$x(t)$ converges in finite time to a vector $x^{*}\in\mathbb{R}^{n}$
belonging to the set $\mathcal{D}_{\varepsilon}$ given by \eqref{eq:finalset},
almost surely. \end{thm} 

\begin{proof}By Proposition~\ref{proposition:persistency_async},
we have \eqref{eq:phi_condition}. Thus, the result follows from Theorem~\ref{Theorem:consensus}.
\end{proof}

We emphasize again that the result does not depend on the frequency
of attacks. In particular, the proposed stochastic communication protocol
allows us to deal with attack scenarios where the jamming is turned
on and off very frequently. 

\subsection{Consensus Under Communication-Aware Attacks \label{subsec:Consensus-Under-Communication-Aw}}

\label{section:wiser_attack}

We now explore an attack strategy where the attacker can sense communication
attempts on the channel and turns the jamming on and off based on
the activity of the agents. Here, we consider a simpler setup where
$\Delta^{i}=\Delta$, $i\in\mathcal{V}$, with $\Delta>0$. In this
setup, $k$th communication attempt time of each agent is in the interval
$[k\Delta,(k+1)\Delta)$, i.e., $t_{k}^{i}\in[k\Delta,(k+1)\Delta)$,
$i\in\mathcal{V}$. Let $\bar{t}_{k}\triangleq\max_{i\in\mathcal{V}}\{t_{k}^{i}\}$. 

We consider an attack strategy where the attacker knows about the
communication protocol as well as $\Delta$. The attacker generates
an attack so that for each interval $[k\Delta,(k+1)\Delta)$,
\begin{enumerate}
\item [i)] the jamming attack starts from $t=k\Delta$,
\item [ii)] the jamming attack continues until time $\bar{t}_{k}$ as long
as Assumption \ref{Assumption:duration_sync} is satisfied.
\end{enumerate}
Under this strategy, for the interval $[k\Delta,(k+1)\Delta)$, the
attacker turns off jamming right after all communication attempts
are blocked. This allows the attacker to preserve energy to be used
later. Furthermore, the attacker can block transmissions among the
agents in a more thorough way than the deterministic strategy. Clearly,
if the attacker wants to interfere at the randomly chosen communication
attempt times for sure, he has to keep jamming until $\bar{t}_{k}$
in each interval $[k\Delta,(k+1)\Delta)$. 

To characterize $a_{k},\tau_{k},k\in\mathbb{N}_{0}$, for this strategy,
first let 
\begin{align*}
s_{k} & \triangleq\mathrm{max}\Big\{ s\in[0,\Delta]\,\colon\\
 & \quad\,\,|\overline{\mathcal{A}}(\tau,k\Delta)|+s\leq\kappa+\rho(k\Delta+s-\tau),\,\,\tau\in[0,k\Delta]\Big\},
\end{align*}
 for $k\in\mathbb{N}_{0}$. Note that $s_{k}\in[0,\Delta]$ denotes
the largest duration that a jamming attack starting at $k\Delta$
can last without violating the condition \eqref{eq:attack-duration-condition}
in Assumption \ref{Assumption:duration_sync}. Now, let
\begin{align*}
\underline{s}_{k} & \triangleq\min\{\bar{t}_{k}-k\Delta,s_{k}\},\quad k\in\mathbb{N}_{0}.
\end{align*}
 Observe that $\underline{s}_{k}$ gives the duration of the attack
in the interval $[k\Delta,(k+1)\Delta)$ for this strategy. In particular,
the jamming attack is turned on for $t\in[k\Delta,k\Delta+\underline{s}_{k}]$,
and turned off for $t\in(k\Delta+\underline{s}_{k},(k+1)\Delta)$.
Hence, $a_{k},\tau_{k}$ can be given by 
\begin{align}
a_{k} & =k\Delta,\quad\tau_{k}=\underline{s}_{k},\quad k\in\mathbb{N}_{0}.\label{eq:wiser-h-tau-1}
\end{align}
Consequently, the set of time instants where communication is not
possible in the interval $[k\Delta,(k+1)\Delta)$ is then given by
the set $\mathcal{A}_{k}\triangleq[a_{k},a_{k}+\tau_{k}]=[k\Delta,k\Delta+\underline{s}_{k}]$.
We remark that the communication-aware jamming described by \eqref{eq:wiser-h-tau-1}
satisfies \eqref{eq:attack-duration-condition} in Assumption~\ref{Assumption:duration_sync}
by construction. 

\begin{figure}[t]
\begin{centering}
\includegraphics[width=0.85\columnwidth]{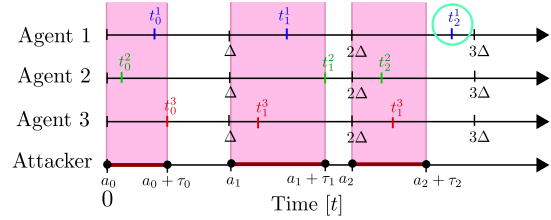} \vskip -9pt
\par\end{centering}
\caption{Attack times for the communication-aware attack strategy. }
\label{figure:wiser_attack_pattern} 
\end{figure}

To illustrate the properties of the communication-aware attack strategy,
we show an example attack scenario in Figure~\ref{figure:wiser_attack_pattern}.
Here, the attacker is able to block all communications that are attempted
in the interval $[0,\Delta)$ by jamming the network between times
$0$ and $\bar{t}_{0}=\max_{i\in\mathcal{V}}t_{0}^{i}$. After blocking,
the attacker turns off jamming and waits until the next interval $[\Delta,2\Delta)$.
In the interval $[\Delta,2\Delta)$, the duration of the attack is
relatively large, because agent $2$ attempts communication towards
the end of the interval. The attacker has to use large energy resources
for this interval. As a result, the attacker cannot conduct an attack
with very long duration in the interval $[2\Delta,3\Delta)$, since
Assumption~\ref{Assumption:duration_sync} holds only for a short
duration of length $\tau_{2}$. And after that, the attacker has to
turn off jamming to save resources. Therefore, in the interval $[2\Delta,3\Delta)$,
it happens that agent $1$ can successfully communicate after the
jamming is turned off. As a result, we have $\varphi_{0}^{1}=0$,
$\varphi_{1}^{1}=0$, $\varphi_{2}^{1}=1$. 

We show in the following that the agents achieve consensus under the
communication-aware attack strategy given in \eqref{eq:wiser-h-tau-1}
by using our proposed stochastic communication protocol. 

Now, consider $\gamma^{i}\in\mathbb{N}$, $\hat{\Delta}^{i}>0$, and
$\hat{\varphi}_{k}^{i}\in\{0,1\}$, $k\in\mathbb{N}_{0}$, $i\in\mathcal{V}$,
given by \eqref{eq:gamma-def}--\eqref{eq:phi_hat_def} with $\Delta^{i}=\Delta$.
Note that in this subsection, we have $\gamma^{i}=\gamma^{j}$, $\hat{\Delta}^{i}=\hat{\Delta}^{j}$
for all $i,j\in\mathcal{V}$, since $\Delta^{i}=\Delta^{j}$, $i,j\in\mathcal{V}$.
We can thus simplify the notation by setting 
\begin{align}
\gamma & \triangleq\gamma^{1},\quad\hat{\Delta}\triangleq\hat{\Delta}^{1}.\label{eq:gamma-Delta-phi-def}
\end{align}

The analysis of consensus under communication-aware jamming attacks
is quite different from the case with deterministic jamming attacks.
Here, we utilize a filtration representing the timing of the attacks
and communication attempt instants. In particular, we consider the
filtration $\{\mathcal{H}_{k}^{i}\}_{k\in\mathbb{N}_{0}}$, where
$\mathcal{H}_{k}^{i}$ denotes the $\sigma$-algebra generated by
the random variables $a_{0},a_{1},\ldots a_{(k+1)\gamma-1},\tau_{0},\tau_{1},\ldots,\tau_{(k+1)\gamma-1}$,
and $t_{0}^{i},t_{1}^{i},\ldots,t_{(k+1)\gamma-1}^{i}$. Notice that
$\varphi_{j}^{i}$, $j\in\{0,\ldots,(k+1)\gamma-1\}$, are $\mathcal{H}_{k}^{i}$-measurable
random variables, because $\varphi_{j}^{i}$ is determined by $a_{j}$,
$\tau_{j}$, and $t_{j}^{i}$. Consequently, $\hat{\varphi}_{j}^{i}$,
$j\in\{0,\ldots,k\}$, are also $\mathcal{H}_{k}^{i}$-measurable.
In the statement of the results below, in addition to $\{\mathcal{H}_{k}^{i}\}_{k\in\mathbb{N}_{0}}$,
we also use the $\sigma$-algebra $\mathcal{H}_{-1}^{i}$, which we
define as $\mathcal{H}_{-1}^{i}\triangleq\{\emptyset,\Omega\}$. 

In what follows our main objective is to show that the agents can
communicate with their neighbors infinitely many times in the long
run satisfying \eqref{eq:phi_condition}, even though the network
faces communication-aware jamming attacks described in \eqref{eq:wiser-h-tau-1}.
We show this by establishing several key results. First, we investigate
the probability of successful communications in the intervals $[k\hat{\Delta},(k+1)\hat{\Delta})$,
$k\in\mathbb{N}_{0}$. The following result provides a positive lower-bound
for the conditional probability of a successful communication in $[k\hat{\Delta},(k+1)\hat{\Delta})$
given $\mathcal{H}_{k-1}^{i}$ (i.e., the information on all previous
intervals).

\begin{lem} \label{Lemma:Positive-prob} Consider the stochastic
communication protocol in Definition~\ref{Def:Communication}. For
the attacks given by \eqref{eq:wiser-h-tau-1}, we have 
\begin{equation}
\mathbb{P}[\hat{\varphi}_{k}^{i}=1\,|\,\mathcal{H}_{k-1}^{i}]\geq2q^{\gamma},\quad k\in\mathbb{N}_{0},\quad i\in\mathcal{V},\label{eq:Phi-last-prob-1}
\end{equation}
 where 
\begin{align}
 & q\triangleq\frac{\widetilde{\Delta}}{\Delta},\quad\widetilde{\Delta}\triangleq\cfrac{(1-\rho)(\hat{\Delta}-\underline{\Delta})}{\gamma+1},\quad\underline{\Delta}\triangleq\frac{\kappa}{1-\rho}.\label{eq:qdef}
\end{align}
 \end{lem} 

 The proof of Lemma~\ref{Lemma:Positive-prob} is given in  Appendix~\ref{Appendix-Proof-Positive-prob}. 

In \eqref{eq:Phi-last-prob-1}, the conditional probability term $\mathbb{P}[\hat{\varphi}_{k}^{i}=1\,|\,\mathcal{H}_{k-1}^{i}]$
is an $\mathcal{H}_{k-1}^{i}$-measurable random variable. Furthermore,
its expectation gives the probability of having $\hat{\varphi}_{k}^{i}=1$,
i.e., $\mathbb{P}[\hat{\varphi}_{k}^{i}=1]=\mathbb{E}[\mathbb{P}[\hat{\varphi}_{k}^{i}=1\,|\,\mathcal{H}_{k-1}^{i}]]$.
Hence, Lemma~\ref{Lemma:Positive-prob} implies $\mathbb{P}[\hat{\varphi}_{k}^{i}]>0$,
$k\in\mathbb{N}_{0}$. In other words, for each interval $[k\hat{\Delta},(k+1)\hat{\Delta})$,
our stochastic communication protocol guarantees a positive probability
for a successful communication.

In the proof of Lemma~\ref{Lemma:Positive-prob}, we consider the
interval $[k\hat{\Delta}^{i},(k+1)\hat{\Delta})$ that is composed
of $\gamma$ number of $\Delta$-length intervals. In each of these
$\Delta$-length intervals, agent $i$ attempts to communicate once.
In our approach, we find a lower bound for $\mathbb{P}[\hat{\varphi}_{k}^{i}=1\,|\,\mathcal{H}_{k-1}^{i}]$
(the conditional probability that at least $1$ out of $\gamma$ communication
attempts is successful). This is done by computing a lower bound for
$\mathbb{P}[\varphi_{(k+1)\gamma-1}^{i}=1|\mathcal{H}_{k-1}^{i}]$,
which is the conditional probability that the \emph{last} attempt
is successful. The key method in deriving this bound is the construction
of the event $G_{k}$ given in \eqref{eq:Gk-def}. Here, $G_{k}$
is the event that the first $\gamma-1$ number of communication attempts
of agent $i$ happen in the last $\widetilde{\Delta}$ units of time
in their respective $\Delta$-length intervals. If $G_{k}$ happens,
then it means that the attacker needs to use sufficiently large jamming
resources to block those first $\gamma-1$ attempts. As a result,
the attacker would not have enough resources left to guarantee blocking
the last attempt. This allows us to compute a lower bound of $\mathbb{P}[\{\varphi_{(k+1)\gamma-1}^{i}=1\}\cap G_{k}|\mathcal{H}_{k-1}^{i}]$.
We then use the inequality $\mathbb{P}[\varphi_{(k+1)\gamma-1}^{i}=1|\mathcal{H}_{k-1}^{i}]\geq\mathbb{P}[\{\varphi_{(k+1)\gamma-1}^{i}=1\}\cap G_{k}|\mathcal{H}_{k-1}^{i}]$
to arrive at the result \eqref{eq:Phi-last-prob-1}. This result is
crucial in proving the following lemma.

\begin{lem} \label{LemmaIntersection} Consider the attack strategy
described by \eqref{eq:wiser-h-tau-1}. Under the stochastic communication
protocol in Definition~\ref{Def:Communication}, we have 
\begin{align}
\mathbb{P}[\bigcap_{k=0}^{N-1}\{\hat{\varphi}_{k}^{i}=\overline{\varphi}_{k+1}\}] & \leq\prod_{j=1}^{N}\big(1-2q^{\gamma}(1-\overline{\varphi}_{j})\big),\label{eq:IntersectionLemmaResult}
\end{align}
for $\overline{\varphi}_{1},\overline{\varphi}_{2},\ldots,\overline{\varphi}_{N}\in\{0,1\}$
and $N\in\mathbb{N}$. 

\end{lem}

\vskip 5pt 

 The proof of Lemma~\ref{LemmaIntersection} is given in Appendix~\ref{Appendix-Proof-Intersection}. 

Lemma~\ref{LemmaIntersection} provides an upper bound for the probability
of the event that the random variables $\hat{\varphi}_{0}^{i},\hat{\varphi}_{1}^{i},\ldots,\hat{\varphi}_{N-1}^{i}$
take the particular values $\overline{\varphi}_{1},\overline{\varphi}_{2},\ldots,\overline{\varphi}_{N}\in\{0,1\}$,
respectively. This result is important because the upper-bound can
be given in terms of the scalar $q$, which depends on $\rho$ and
$\kappa$ characterizing the attacker's capabilities as well as the
parameter $\Delta$ of the communication protocol. Notice that if
the sequence $\overline{\varphi}_{1},\overline{\varphi}_{2},\ldots,\overline{\varphi}_{N}$
is formed of $m$ number of $1$s and $N-m$ number of $0$s, then
the probability bound in \eqref{eq:IntersectionLemmaResult} is given
by $(1-2q^{\gamma})^{N-m}$. The following result is built upon this
observation. 

\begin{prop} \label{PropositionLowerBound} Consider the attack strategy
described by \eqref{eq:wiser-h-tau-1}. Under the stochastic communication
protocol in Definition~\ref{Def:Communication}, we have 
\begin{align}
\mathbb{P}[\sum_{k=0}^{N-1}\hat{\varphi}_{k}^{i}\geq M] & \geq1-\sum_{m=0}^{M-1}\frac{N!}{m!(N-m)!}(1-2q^{\gamma})^{N-m},\label{eq:LowerBoundResult}
\end{align}
 for all $M\in\{0,1,\ldots,N\}$ and $N\in\mathbb{N}$. 

\end{prop}

\begin{proof}First, we obtain 
\begin{align}
 & \mathbb{P}[\sum_{k=0}^{N-1}\hat{\varphi}_{k}^{i}\geq M]=1-\mathbb{P}[\bigcup_{m=0}^{M-1}\{\sum_{k=0}^{N-1}\hat{\varphi}_{k}^{i}=m\}]\nonumber \\
 & \quad\geq1-\sum_{m=0}^{M-1}\mathbb{P}[\sum_{k=0}^{N-1}\hat{\varphi}_{k}^{i}=m].\label{eq:FirstIneqWithMandm}
\end{align}
Now, let $\Pi_{N,m}^{i}\triangleq\{\overline{\varphi}\in\{0,1\}^{N}\colon\overline{\varphi}^{\mathrm{T}}\overline{\varphi}=m\}$
for $m\in\{0,1,\ldots,M\}$ and $N\in\{M,M+1,\ldots\}$. Notice that
\begin{align}
\mathbb{P}[\sum_{k=0}^{N-1}\hat{\varphi}_{k}^{i}=m] & =\mathbb{P}[\cup_{\overline{\varphi}\in\Pi_{N,m}^{i}}\cap_{k=0}^{N-1}\{\hat{\varphi}_{k}^{i}=\overline{\varphi}_{k+1}\}]\nonumber \\
 & \leq\sum_{\overline{\varphi}\in\Pi_{N,m}^{i}}\mathbb{P}[\cap_{k=0}^{N-1}\{\hat{\varphi}_{k}^{i}=\overline{\varphi}_{k+1}\}].\label{eq:SumOverPi}
\end{align}
By using Lemma~\ref{LemmaIntersection} we obtain from \eqref{eq:SumOverPi}
that 
\begin{align}
\mathbb{P}[\sum_{k=0}^{N-1}\hat{\varphi}_{k}^{i}=m] & \leq\sum_{\overline{\varphi}\in\Pi_{N,m}^{i}}\prod_{j=1}^{N}\big(1-2q^{\gamma}(1-\overline{\varphi}_{j})\big).\label{eq:SumOfProducts}
\end{align}
Note that $\prod_{j=1}^{N}\big(1-2q^{\gamma}(1-\overline{\varphi}_{j})\big)=(1-2q^{\gamma})^{N-m}$
for $\overline{\varphi}\in\Pi_{N,m}^{i}$. Furthermore, the set $\Pi_{N,m}^{i}$
has $\frac{N!}{m!(N-m)!}$ elements. Therefore, it follows from \eqref{eq:SumOfProducts}
that 
\begin{align}
\mathbb{P}[\sum_{k=0}^{N-1}\hat{\varphi}_{k}^{i}=m] & \leq\sum_{\overline{\varphi}\in\Pi_{N,m}^{i}}(1-2q^{\gamma})^{N-m}\nonumber \\
 & =\frac{N!}{m!(N-m)!}(1-2q^{\gamma})^{N-m}.\label{eq:FinalIneqForProbm}
\end{align}
 Finally, by using \eqref{eq:FirstIneqWithMandm} and \eqref{eq:FinalIneqForProbm},
we arrive at \eqref{eq:LowerBoundResult}. \end{proof}

Proposition~\ref{PropositionLowerBound} provides a lower bound of
the probability that agent $i$ can communicate with its neighbors
at least $M$ times during the interval $[0,N\hat{\Delta})$. Notice
that as $N$ approaches $\infty$, this lower bound approaches $1$. 

\begin{thm} \label{Theorem:Wiser-phi} Consider the attack strategy
described by \eqref{eq:wiser-h-tau-1}. Under the stochastic communication
protocol in Definition~\ref{Def:Communication}, the equality in
\eqref{eq:phi_condition} holds. \end{thm} \vskip 10pt

\begin{proof} Our initial goal is to show 
\begin{align}
\mathbb{P}[\sum_{k=0}^{\infty}\hat{\varphi}_{k}^{i}\geq M] & =1,\quad M\in\mathbb{N}_{0},\quad i\in\mathcal{V}.\label{eq:varphihatineq}
\end{align}
 To this end, first let $A_{N}\triangleq\{\omega\in\Omega\colon\sum_{k=0}^{N-1}\hat{\varphi}_{k}^{i}\geq M\}$,
$N\in\mathbb{N}$. Notice that $\mathbb{P}[A_{N}]=0$ for $N<M$.
For $N\geq M$, Proposition~\ref{PropositionLowerBound} implies
\begin{align}
\mathbb{P}[A_{N}] & \geq1-\sum_{m=0}^{M-1}\frac{N!}{m!(N-m)!}(1-2q^{\gamma})^{N-m}.\label{eq:PAnLowerBound}
\end{align}
 Since $1-2q^{\gamma}<1$, it follows from \eqref{eq:PAnLowerBound}
that $\lim_{N\to\infty}\mathbb{P}[A_{N}]=1$. The events $A_{N}$,
$N\in\mathbb{N}$, satisfy $A_{N}\subseteq A_{N+1}$. Hence, by the
\emph{monotone-convergence theorem for sets} (see Section~1.10 in
\citeasnoun{williams2010probability}), 
\begin{align}
\mathbb{P}[\sum_{k=0}^{\infty}\hat{\varphi}_{k}^{i} & \geq M]=\mathbb{P}[\lim_{N\to\infty}A_{N}]=\lim_{N\to\infty}\mathbb{P}[A_{N}]=1.\label{eq:MonotoneConvergenceTheoremResult}
\end{align}
Finally, since $\mathbb{P}[\sum_{k=0}^{\infty}\varphi_{k}^{i}\geq M]\geq\mathbb{P}[\sum_{k=0}^{\infty}\hat{\varphi}_{k}^{i}\geq M]$,
it follows from \eqref{eq:MonotoneConvergenceTheoremResult} that
\eqref{eq:phi_condition} holds. \end{proof}

Theorem~\ref{Theorem:Wiser-phi} shows that the agents can communicate
with their neighbors infinitely many times in the long run, even though
the network is attacked by an attacker that follows the communication-aware
attack strategy described in \eqref{eq:wiser-h-tau-1}. The next theorem
is the main result for the multi-agent system under communication-aware
attacks.

\begin{thm}\label{Theorem:consensus-wiser} Consider the multi-agent
system \eqref{eq:scalar-dynamics}, \eqref{eq:control-input} with
$T^{i}\in(0,\min\{\frac{\varepsilon}{2d^{i}},\Delta^{i}\})$ where
$\varepsilon>0$. For the attack strategy described by \eqref{eq:wiser-h-tau-1},
$x(t)$ converges in finite time to a vector $x^{*}\in\mathbb{R}^{n}$
belonging to the set $\mathcal{D}_{\varepsilon}$ given by \eqref{eq:finalset},
almost surely. \end{thm} 

\begin{proof}By Theorem~\ref{Theorem:Wiser-phi}, we have \eqref{eq:phi_condition}.
Consequently, the result follows from Theorem~\ref{Theorem:consensus}.
\end{proof}

So far we considered the consensus problem under both deterministic
attacks and communication-aware attacks. In both cases, the randomness
in the communication attempt times is the key property that enables
consensus regardless of the frequency of jamming. A difference is
that the attacker following the communication-aware attack strategy
can sense the network activity and switch off the jamming attack right
after blocking a communication attempt. This allows the attacker to
preserve energy. This is further illustrated through numerical examples
in the next section. 

\begin{figure}[t]
\begin{centering}
\includegraphics[width=0.55\columnwidth]{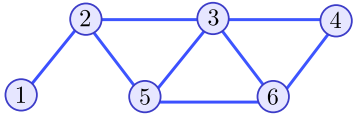} \vskip -8pt
\par\end{centering}
\caption{Communication topology of the multi-agent system. }
\label{figure:network_topology} 
\end{figure}

\section{Numerical Examples \label{sec:Numerical-Examples} }


In this section, we illustrate our results for the multi-agent system
with $n=6$ agents whose topology is shown in Figure~\ref{figure:network_topology}. 

\subsection{Deterministic Attacks \label{subsec:Deterministic-Attacks}}

We first consider a deterministic attack scenario where the strategy
of the attacker satisfies Assumptions~\ref{Assumption:duration_sync}
and \ref{assumption:deterministic_sync} with $\kappa=0.2$, $\rho=0.8$.
We utilize our proposed stochastic communication protocol with $\Delta^{i}=0.001+0.0001(i+1)$,
$i\in\mathcal{V}$. For the control laws \eqref{eq:control-input},
we choose $T^{i}=\Delta^{i}/1.01$, which satisfy $T^{i}\in(0,\min\{\frac{\varepsilon}{2d^{i}},\Delta^{i}\})$
with $\varepsilon=0.02$. Since Proposition~\ref{proposition:persistency_async}
implies \eqref{eq:phi_condition}, it follows from Theorem~\ref{Theorem:consensus-deterministic}
that the multi-agent system achieves consensus. 

In the top part of Figure~\ref{figure:agent-states-deterministic-attack},
we show sample paths of agent states under jamming attacks with low
frequency. We see that consensus is achieved around the time $t=3.85$.
 Each agent $i$ attempts to communicate once at a random time instant
at every $\Delta^{i}$ units of time. The agents keep their states
constant during long jamming intervals.

\begin{figure}[t]
\begin{centering}
\includegraphics[width=0.9\columnwidth]{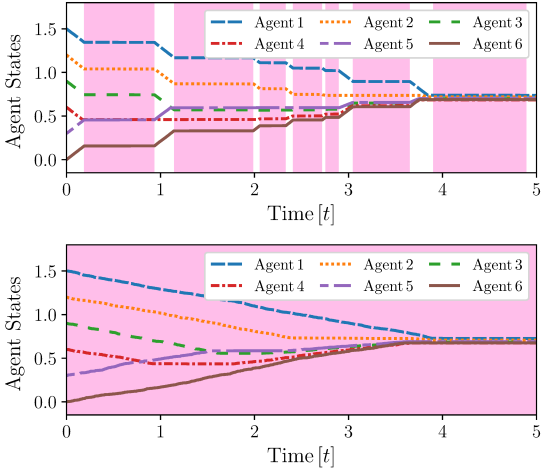}
\vskip -15pt
\par\end{centering}
\caption{Evolution of agent states under deterministic attacks with low frequency
(top) and high frequency (bottom) settings. }
\label{figure:agent-states-deterministic-attack} 
\end{figure}

The attack depicted in the top part of Figure~\ref{figure:agent-states-deterministic-attack}
is of low frequency, as the jamming is turned on and off only $7$
times during the interval $[0,5]$. We also consider a high frequency
case in the bottom plot of Figure~\ref{figure:agent-states-deterministic-attack},
where jamming is turned on and off $7935$ times during the interval
$[0,5]$, but the agent communication attempt times are the same as
those in the top plot. Also in this case, the agents reach the consensus
set $\mathcal{D}_{\varepsilon}$ around the time $t=3.91$. Both the
low and the high frequency attacks in Figure~\ref{figure:agent-states-deterministic-attack}
are generated randomly and independently of the communication attempt
times of the agents. Through repeated simulations, we also observe
that consensus is reached around the same time. 

Next, we consider periodically generated jamming attacks with 
\begin{align}
a_{k} & \triangleq\frac{k}{\sigma}+\frac{(1-\rho)}{\sigma},\quad\tau_{k}\triangleq\frac{\rho}{\sigma},\quad k\in\mathbb{N}_{0},\label{eq:example_attacks}
\end{align}
where $\sigma>0$ denotes the frequency of attacks (i.e., the number
of attack intervals in $1$ unit of time). Moreover, $\rho>0$ indicates
the ratio of the duration of attacks in each period. For each $\rho\in\{0.2,0.5,0.8\}$
and $\sigma\in\{10^{1},10^{3},10^{5}\}$ we repeat the simulation
$50$ times. For each simulation $j\in\{1,\ldots,50\}$, we calculate
$t_{\mathrm{C}}(j)\triangleq\inf\{t\colon x^{i}(t)\in\mathcal{D}_{\varepsilon},i\in\mathcal{V}\}$,
which is the time agents reach consensus. Then we obtain their mean
$m_{\mathrm{C}}>0$ and standard deviation $s_{\mathrm{C}}>0$. 

\begin{table}[t]
\caption{Mean $m_{\mathrm{C}}$ and standard deviation $s_{\mathrm{C}}$ of
consensus times for different values of $\rho$ and $\sigma$ in deterministic
attacks \eqref{eq:example_attacks}. }
\label{TableForMcandSc}

\renewcommand{\arraystretch}{0.8} \vskip 5pt

{\centering \fontsize{7}{11.52}\selectfont  

\begin{tabular}{ccccccccc}
\toprule 
 & \multicolumn{2}{c}{$\sigma=10^{1}$} &  & \multicolumn{2}{c}{$\sigma=10^{3}$} &  & \multicolumn{2}{c}{$\sigma=10^{5}$}\tabularnewline
\cmidrule{2-3} \cmidrule{3-3} \cmidrule{5-6} \cmidrule{6-6} \cmidrule{8-9} \cmidrule{9-9} 
$\rho$ & $m_{\mathrm{C}}$ & $s_{\mathrm{C}}$  &  & $m_{\mathrm{C}}$  & $s_{\mathrm{C}}$ &  & $m_{\mathrm{C}}$  & $s_{\mathrm{C}}$\tabularnewline
\midrule
\midrule 
$0.2$ & $1.143$ & $0.005$ &  & $1.108$ & $0.015$ &  & $1.110$ & $0.013$\tabularnewline
\midrule 
$0.5$ & $1.822$ & $0.006$ &  & $1.663$ & $0.029$ &  & $1.682$ & $0.035$\tabularnewline
\midrule 
$0.8$ & $4.532$ & $0.035$ &  & $3.962$ & $0.086$ &  & $4.010$ & $0.107$\tabularnewline
\bottomrule
\end{tabular}

} 
\end{table}

Table~\ref{TableForMcandSc} indicates that increasing the ratio
$\rho$ of the attack duration allows the attacker to delay the consensus.
On the other hand, consensus time is not influenced a lot by how frequent
the attacks are. For each value of $\rho$, mean consensus time $m_{\mathrm{C}}$
is similar under all attack frequency settings $\sigma=10^{1},$ $10^{3}$,
$10^{5}$. Furthermore, consensus times are finite in all simulations
and they do not show large deviation (i.e., $s_{\mathrm{C}}$ is small)
in all cases. The cases with $\rho=0.8$ indicate that periodic attacks
and the attack timings shown in Figure~\ref{figure:agent-states-deterministic-attack}
do not differ much in their effects on consensus times. 

\subsection{Communication-Aware Attacks \label{subsec:Communication-Aware-Attacks}}

Next, we consider the scenario where the attacker follows the communication-aware
attack strategy of \eqref{eq:wiser-h-tau-1} with the same parameters
$\kappa=0.2$ and $\rho=0.8$ as in Section~\ref{subsec:Deterministic-Attacks}. 

In this scenario, the intervals for the communication are selected
as $\Delta^{i}=\Delta=0.001$, $i\in\mathcal{V}$. Similar to the
deterministic case discussed above, for the control law \eqref{eq:control-input},
we choose $T^{i}=\Delta/1.01$, $i\in\mathcal{V}$, which satisfy
$T^{i}\in(0,\min\{\frac{\varepsilon}{2d^{i}},\Delta^{i}\})$ with
$\varepsilon=0.02$. Furthermore, Theorem~\ref{Theorem:Wiser-phi}
implies that \eqref{eq:phi_condition} holds. Therefore, it follows
from Theorem~\ref{Theorem:consensus-wiser} that the multi-agent
system with the stochastic communication protocol achieves consensus. 

We show the evolution of the agent states in Figure~\ref{figure:agent-states-wiser-attack}.
Notice that every communication attempt in the interval $[0,3.18]$
is blocked by the attacker. However, the attacker's energy resources
eventually become not sufficient. We observe in the enlarged plot
in the bottom part of Figure~\ref{figure:agent-states-wiser-attack}
that some of the communication attempts cannot be blocked by the attacker
and the agents eventually achieve consensus. 

\begin{figure}[t]
\begin{centering}
\includegraphics[width=0.9\columnwidth]{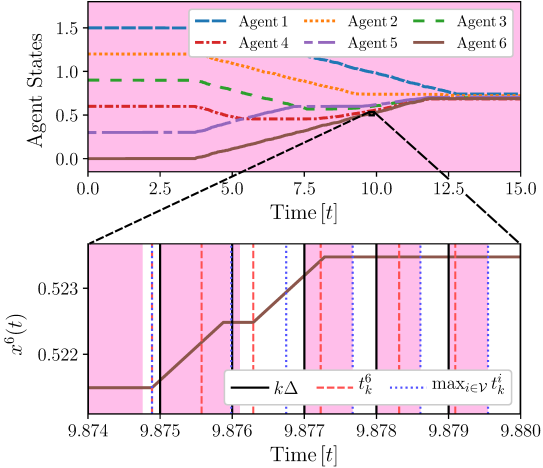}
\vskip -15pt
\par\end{centering}
\caption{Evolution of agent states under communication-aware attacks. }
\label{figure:agent-states-wiser-attack} 
\end{figure}

Even though the value of $\rho$ is the same with $\rho=0.8$ as in
the deterministic attacks case of the previous example, the communication-aware
attacks can be more malicious in the sense that they can delay consensus
(compare Figures~\ref{figure:agent-states-deterministic-attack}
and \ref{figure:agent-states-wiser-attack}). To further investigate
how $\rho$ and $\kappa$ affect the consensus time, we run simulations
with different values of $\rho$ and $\kappa$ but with the same communication
attempt times used for constructing Figure~\ref{figure:agent-states-wiser-attack}.
We observe in Figure~\ref{figure:rho-effect} that consensus time
$t_{\mathrm{C}}$ increases as $\rho$ increases. The scalar $\kappa\geq0$
also has an effect on the consensus time. In particular, increasing
$\kappa$ delays the consensus, since the duration for continuous
jamming becomes larger. 

\begin{figure}[t]
\begin{centering}
\includegraphics[width=0.9\columnwidth]{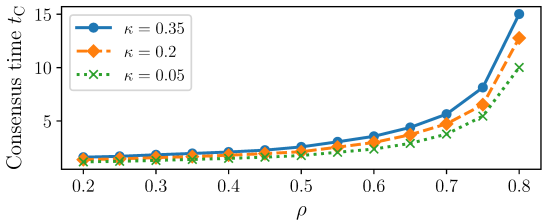}
\vskip -15pt
\par\end{centering}
\caption{Effect of the bound $\rho$ of average attack duration ratio to consensus
time $t_{\mathrm{C}}$ for different values of $\kappa$. }
\label{figure:rho-effect} 
\end{figure}

We remark that in communication-aware attacks, the attacker turns
jamming on and off once in every $\Delta$-length intervals. Hence,
the frequency of attacks is equal to the frequency of communication
attempts. This case is outside the class of attacks considered previously
in \citeasnoun{senejohnny_2015}. On the other hand, the class of
attacks under which our communication protocol allows consensus is
not restricted by the frequency of attacks. Specifically, as long
as the average ratio of the duration of attacks in the long run is
bounded by $\rho<1$, consensus can be achieved. 

\begin{rem} Although the randomized transmission approach expands
the class of attacks under which consensus can be achieved, the number
of communication attempts made by agents can be large compared to
that in the self-triggered communication schemes of \citeasnoun{de_persis_ternary_2013},
\citeasnoun{senejohnny_2015}, and \citeasnoun{senejohnny2017jamming}.
In the self-triggered communication schemes, each agent $i$ computes
its next communication time $t_{k+1}^{i}$ based on neighboring agents'
states at time $t_{k}^{i}$. If the states of neighboring agents are
far from agent $i$'s state, then $t_{k+1}^{i}$ takes a large value.
This approach reduces the number of transmissions and keeps the communication
costs low. In the randomized transmission scheme of this paper, the
number of communication attempts made by agent $i$ in a fixed interval
$[0,t]$ can be reduced by increasing $\Delta^{i}$ in Definition~\ref{Def:Communication}.
However, setting a large value to $\Delta^{i}$ results in slow convergence
towards the consensus set, since each agent $i$ applies control inputs
only for at most $T^{i}$ units of time between each communication
attempt. The approximate consensus time achieved with our control
approach can be larger than the time achieved with the approach of
above-mentioned works, since in our approach, control inputs of agents
may be set to zero between consecutive communication times. 

To reduce the expected total time where the control inputs have zero
value before reaching consensus, $T^{i}$ can be selected close to
$\Delta^{i}$. This also reduces the time to reach consensus. We observe
through simulations that with smaller $T^{i}$ (given by $T^{i}=\Delta^{i}/2.1$)
consensus is achieved at time $t_{\mathrm{C}}=22.92$. On the other
hand, when $T^{i}$ is larger ($T^{i}=\Delta^{i}/1.01$ as in the
setting of this section) consensus is achieved for the same attack
parameters at an earlier time $t_{\mathrm{C}}=12.79$ as shown in
Figure~\ref{figure:agent-states-wiser-attack}. 

We note that a new \emph{hybrid }self-triggered and randomized approach
can be useful in reducing the communication loads while protecting
against a large class of attacks and achieving faster consensus. In
the hybrid approach, the neighboring agents' states obtained at time
$t_{k}^{i}$ can be used by agent $i$ for choosing the length $\Delta_{k+1}^{i}$
of the interval in which the next transmission attempt time $t_{k+1}^{i}$
is randomly selected. \end{rem} 

\section{Conclusion \label{sec:Conclusion}}

We proposed a stochastic communication protocol for multi-agent consensus
under jamming attacks. In this protocol, agents attempt to exchange
information with their neighbors at uniformly distributed random time
instants. We showed that our proposed communication protocol guarantees
consensus as long as the jamming attacks satisfy a certain condition
on the average ratio of their duration. We demonstrated our results
both for a deterministic attack strategy and a communication-aware
attack strategy. 

The analysis in this paper enables a natural extension to the case
with multiple jamming attackers that can attack different links at
different times. In such a problem setting, if the deterministic or
the communication-aware attack for each communication link satisfies
Assumption~\ref{Assumption:duration_sync}, then two agents can communicate
over that link infinitely many times in the long run. This allows
agents to achieve consensus through a modified control law where each
agent can communicate with different neighbors at different times. 

\bibliographystyle{ifac}
\bibliography{references}

\appendix
\section{Proof of Lemma~\ref{Lemma:Positive-prob}}
\label{Appendix-Proof-Positive-prob}
\linespread{0.92}
\setlength{\jot}{2ex}       
\setlength\abovedisplayskip{3.6pt}       
\setlength\belowdisplayskip{3.6pt}

In the interval $[k\hat{\Delta}^{i},(k+1)\hat{\Delta})=[k\gamma\Delta,(k+1)\gamma\Delta)$,
agent $i$ attempts communication with its neighbors for $\gamma$
number of times at time instants $t_{k\gamma}^{i},t_{k\gamma+1}^{i},\ldots,t_{(k+1)\gamma-1}^{i}$.
It follows from the definition of $\hat{\varphi}_{k}^{i}$ in \eqref{eq:phi_hat_def}
and \eqref{eq:gamma-Delta-phi-def} that 
\begin{align}
\mathbb{P}[\hat{\varphi}_{k}^{i}=1|\mathcal{H}_{k-1}^{i}] & \geq\mathbb{P}[\varphi_{(k+1)\gamma-1}^{i}=1|\mathcal{H}_{k-1}^{i}],\label{eq:Key-ineq}
\end{align}
 where the right-hand side represents the conditional probability
of a successful communication at time $t_{(k+1)\gamma-1}^{i}$. Hence,
to prove \eqref{eq:Phi-last-prob-1} it suffices to show
\begin{equation}
\mathbb{P}[\varphi_{(k+1)\gamma-1}^{i}=1|\mathcal{H}_{k-1}^{i}]\geq2q^{\gamma},\quad k\in\mathbb{N}_{0}.\label{eq:last-prob}
\end{equation}

To show \eqref{eq:last-prob}, we first consider the case $\gamma=1$.
In this case, $\hat{\Delta}=\Delta$ and $|\overline{\mathcal{A}}(k\Delta,(k+1)\Delta)|=\tau_{k}<\underline{\Delta}$,
almost surely. Moreover, we have 
\begin{align}
 & \mathbb{P}[\varphi_{(k+1)\gamma-1}^{i}=1|\mathcal{H}_{k-1}^{i}]=\mathbb{P}[\varphi_{k}^{i}=1|\mathcal{H}_{k-1}^{i}]\nonumber \\
 & \quad=\mathbb{P}[t_{k}^{i}\notin\overline{\mathcal{A}}(k\Delta,(k+1)\Delta)|\mathcal{H}_{k-1}^{i}]\nonumber \\
 & \quad\geq\mathbb{P}[\{t_{k}^{i}\geq k\Delta+\underline{\Delta}\}\cap\{\tau_{k}<\underline{\Delta}\}|\mathcal{H}_{k-1}^{i}].\label{eq:FirstLowerBound}
\end{align}
Now, since $\mathbb{P}[\tau_{k}\geq\underline{\Delta}]=0$, we have
$\mathbb{P}[\tau_{k}\geq\underline{\Delta}|\mathcal{H}_{k-1}^{i}]=0$,
almost surely. As a result, $\mathbb{P}[\{t_{k}^{i}\geq k\Delta+\underline{\Delta}\}\cap\{\tau_{k}\geq\underline{\Delta}\}|\mathcal{H}_{k-1}^{i}]\leq\mathbb{P}[\tau_{k}\geq\underline{\Delta}|\mathcal{H}_{k-1}^{i}]=0$.
Hence, 
\begin{align}
 & \mathbb{P}[t_{k}^{i}\geq k\Delta+\underline{\Delta}|\mathcal{H}_{k-1}^{i}]\nonumber \\
 & \,\,=\mathbb{P}[\{t_{k}^{i}\geq k\Delta+\underline{\Delta}\}\cap\{\tau_{k}<\underline{\Delta}\}|\mathcal{H}_{k-1}^{i}]\nonumber \\
 & \,\,\quad+\mathbb{P}[\{t_{k}^{i}\geq k\Delta+\underline{\Delta}\}\cap\{\tau_{k}\geq\underline{\Delta}\}|\mathcal{H}_{k-1}^{i}]\nonumber \\
 & \,\,=\mathbb{P}[\{t_{k}^{i}\geq k\Delta+\underline{\Delta}\}\cap\{\tau_{k}<\underline{\Delta}\}|\mathcal{H}_{k-1}^{i}].\label{eq:ProbEquality}
\end{align}
 By using \eqref{eq:FirstLowerBound} and \eqref{eq:ProbEquality},
we obtain 
\begin{align}
\mathbb{P}[\varphi_{(k+1)\gamma-1}^{i}=1|\mathcal{H}_{k-1}^{i}] & \geq\mathbb{P}[t_{k}^{i}\geq k\Delta+\underline{\Delta}|\mathcal{H}_{k-1}^{i}].\label{eq:FirstLowerBound-1}
\end{align}
 Since, $t_{k}^{i}$ is independent of $\mathcal{H}_{k-1}^{i}$, we
have $\mathbb{P}[t_{k}^{i}\geq k\Delta+\underline{\Delta}|\mathcal{H}_{k-1}^{i}]=\mathbb{P}[t_{k}^{i}\geq k\Delta+\underline{\Delta}]$.
It then follows from \eqref{eq:FirstLowerBound-1} that 
\begin{align}
 & \mathbb{P}[\varphi_{(k+1)\gamma-1}^{i}=1|\mathcal{H}_{k-1}^{i}]\geq\mathbb{P}[t_{k}^{i}\geq k\Delta+\underline{\Delta}]=\frac{\Delta-\underline{\Delta}}{\Delta}\nonumber \\
 & \quad=\frac{\hat{\Delta}-\underline{\Delta}}{\Delta}\geq\frac{(1-\rho)(\hat{\Delta}-\underline{\Delta})}{\Delta}=\frac{2\widetilde{\Delta}}{\Delta}=2q,\label{eq:FirstLowerBound-1-1}
\end{align}
which shows that \eqref{eq:last-prob} holds when $\gamma=1$.

Now, consider the case $\gamma\geq2$. By noting $\widetilde{\Delta}<\Delta$,
we let
\begin{align*}
F_{k} & \triangleq\Big\{\omega\in\Omega\,\colon\,t_{k}^{i}\in\big[(k+1)\Delta-\widetilde{\Delta},(k+1)\Delta\big)\Big\},\,\,k\in\mathbb{N}.
\end{align*}
 Observe that $F_{k}\in\mathcal{F}$ denotes the event that the random
communication attempt time $t_{k}^{i}$ falls on the last $\widetilde{\Delta}$
units of time in the interval $[k\Delta,(k+1)\Delta)$. Consider the
interval $[k\hat{\Delta},(k+1)\hat{\Delta})$. Notice that the communication
attempts in this interval occur at time instants $t_{k\gamma}^{i},t_{k\gamma+1}^{i},\ldots,t_{(k+1)\gamma-1}^{i}$.
Let the events $G_{k}\in\mathcal{F}$, $k\in\mathbb{N}_{0}$, be given
by 
\begin{align}
G_{k} & \triangleq F_{k\gamma}\cap F_{k\gamma+1}\cap\cdots\cap F_{(k+1)\gamma-2},\quad k\in\mathbb{N}_{0}.\label{eq:Gk-def}
\end{align}
It follows that 
\begin{align}
 & \mathbb{P}[\varphi_{(k+1)\gamma-1}^{i}=1|\mathcal{H}_{k-1}^{i}]\nonumber \\
 & \quad=\mathbb{P}[\{\varphi_{(k+1)\gamma-1}^{i}=1\}\cap G_{k}|\mathcal{H}_{k-1}^{i}]\nonumber \\
 & \quad\quad+\mathbb{P}[\{\varphi_{(k+1)\gamma-1}^{i}=1\}\cap G_{k}^{\mathrm{c}}|\mathcal{H}_{k-1}^{i}]\nonumber \\
 & \quad\geq\mathbb{P}[\{\varphi_{(k+1)\gamma-1}^{i}=1\}\cap G_{k}|\mathcal{H}_{k-1}^{i}].\label{eq:probability-trick-ahmet}
\end{align}
In the remainder of the proof, we will show 
\begin{align}
\mathbb{P}[\{\varphi_{(k+1)\gamma-1}^{i}=1\}\cap G_{k}|\mathcal{H}_{k-1}^{i}] & \geq2q\mathbb{P}[G_{k}|\mathcal{H}_{k-1}^{i}]\label{eq:FinalIneq1}
\end{align}
 and 
\begin{align}
\mathbb{P}[G_{k}|\mathcal{H}_{k-1}^{i}] & =\mathbb{P}[G_{k}]=q^{\gamma-1}.\label{eq:FinalIneq2}
\end{align}
We will then use \eqref{eq:probability-trick-ahmet}--\eqref{eq:FinalIneq2}
to show \eqref{eq:last-prob}. 

To establish \eqref{eq:FinalIneq1}, we first simplify the presentation
and define the time instants $b_{k}\triangleq k\gamma\Delta$, $c_{k}\triangleq(k+1)\gamma\Delta-\Delta$,
and $d_{k}\triangleq(k+1)\gamma\Delta$, for $k\in\mathbb{N}_{0}$.
Observe that $[b_{k},c_{k})$ gives the union of the first $\gamma-1$
number of $\Delta$-length intervals in $[k\hat{\Delta},(k+1)\hat{\Delta})$,
and moreover, $[c_{k},d_{k})$ gives the last $\Delta$-length interval.
Hence, 
\begin{align}
 & \mathbb{P}[\{\varphi_{(k+1)\gamma-1}^{i}=1\}\cap G_{k}|\mathcal{H}_{k-1}^{i}]\nonumber \\
 & \quad=\mathbb{P}[\{t_{(k+1)\gamma-1}^{i}\notin\overline{\mathcal{A}}(c_{k},d_{k})\}\cap G_{k}|\mathcal{H}_{k-1}^{i}]\nonumber \\
 & \quad\geq\mathbb{P}[\{t_{(k+1)\gamma-1}^{i}>d_{k}-2\widetilde{\Delta}\}\nonumber \\
 & \quad\qquad\cap\{|\overline{\mathcal{A}}(c_{k},d_{k})|\leq\Delta-2\widetilde{\Delta}\}\cap G_{k}|\mathcal{H}_{k-1}^{i}]].\label{eq:t-more-than-first}
\end{align}
By noting that $[k\gamma\Delta,(k+1)\gamma\Delta)=[b_{k},d_{k})=[b_{k},c_{k})\cup[c_{k},d_{k})$,
we obtain $|\overline{\mathcal{A}}(c_{k},d_{k})|=|\overline{\mathcal{A}}(b_{k},d_{k})|-|\overline{\mathcal{A}}(b_{k},c_{k})|$.
It then follows from Assumption \ref{Assumption:duration_sync} that
\begin{align}
|\overline{\mathcal{A}}(c_{k},d_{k})| & \leq\kappa+\rho\hat{\Delta}-|\overline{\mathcal{A}}(b_{k},c_{k})|.\label{eq:assum-ineq}
\end{align}
Noting that $2\widetilde{\Delta}<\Delta$, we use \eqref{eq:assum-ineq}
to show that the events $\{|\overline{\mathcal{A}}(c_{k},d_{k})|\leq\Delta-2\widetilde{\Delta}\}\in\mathcal{F}$
and $\{|\overline{\mathcal{A}}(b_{k},c_{k})|\geq(\gamma-1)(\Delta-\widetilde{\Delta})\}\in\mathcal{F}$
satisfy
\begin{align}
 & \{|\overline{\mathcal{A}}(c_{k},d_{k})|\leq\Delta-2\widetilde{\Delta}\}\nonumber \\
 & \quad\supseteq\{\kappa+\rho\hat{\Delta}-|\overline{\mathcal{A}}(b_{k},c_{k})|\leq\Delta-2\widetilde{\Delta}\}\nonumber \\
 & \quad=\{|\overline{\mathcal{A}}(b_{k},c_{k})|\geq\kappa+\rho\hat{\Delta}-\Delta+2\widetilde{\Delta}\}\nonumber \\
 & \quad=\{|\overline{\mathcal{A}}(b_{k},c_{k})|\geq(\gamma-1)(\Delta-\widetilde{\Delta})\}.\label{eq:event-relation}
\end{align}
 As a consequence of \eqref{eq:t-more-than-first} and \eqref{eq:event-relation},
we obtain 
\begin{align*}
 & \mathbb{P}[\{\varphi_{(k+1)\gamma-1}^{i}=1\}\cap G_{k}|\mathcal{H}_{k-1}^{i}]\\
 & \quad\geq\mathbb{P}[\{t_{(k+1)\gamma-1}^{i}>d_{k}-2\widetilde{\Delta}\}\\
 & \quad\qquad\cap\{|\overline{\mathcal{A}}(b_{k},c_{k})|\geq(\gamma-1)(\Delta-\widetilde{\Delta})\}\cap G_{k}|\mathcal{H}_{k-1}^{i}].
\end{align*}
Thus, since $t_{(k+1)\gamma-1}^{i}$ is independent of $\overline{\mathcal{A}}(b_{k},c_{k})$,
$G_{k}$, and $\mathcal{H}_{k-1}^{i}$, we obtain
\begin{align}
 & \mathbb{P}[\{\varphi_{(k+1)\gamma-1}^{i}=1\}\cap G_{k}|\mathcal{H}_{k-1}^{i}]\nonumber \\
 & \quad\geq\mathbb{P}[t_{(k+1)\gamma-1}^{i}>d_{k}-2\widetilde{\Delta}]\nonumber \\
 & \quad\quad\cdot\mathbb{P}[\{|\overline{\mathcal{A}}(b_{k},c_{k})|\geq(\gamma-1)(\Delta-\widetilde{\Delta})\}\cap G_{k}|\mathcal{H}_{k-1}^{i}]\nonumber \\
 & \quad=2q\mathbb{P}[\{|\overline{\mathcal{A}}(b_{k},c_{k})|\geq(\gamma-1)(\Delta-\widetilde{\Delta})\}\cap G_{k}|\mathcal{H}_{k-1}^{i}].\label{eq:first-ineq}
\end{align}
Here, we have $G_{k}\subseteq\{\omega\in\Omega\colon|\overline{\mathcal{A}}(b_{k},c_{k})|\geq(\gamma-1)(\Delta-\widetilde{\Delta})\}$
and hence $\{|\overline{\mathcal{A}}(b_{k},c_{k})|\geq(\gamma-1)(\Delta-\widetilde{\Delta})\}\cap G_{k}=G_{k}$.
This is because, for every outcome $\omega^{*}\in G_{k}$, the first
$\gamma-1$ communication attempts of agent $i$ happen in the last
$\widetilde{\Delta}$ units of time in their respective intervals,
and thus by \eqref{eq:wiser-h-tau-1}, the total duration of the attacks
in the interval $\big[b_{k},c_{k}\big)$ is at least $(\gamma-1)(\Delta-\widetilde{\Delta})$
units of time, implying $\omega^{*}\in\{\omega\in\Omega\colon|\overline{\mathcal{A}}(b_{k},c_{k})|\geq(\gamma-1)(\Delta-\widetilde{\Delta})\}$.
Using $\{|\overline{\mathcal{A}}(b_{k},c_{k})|\geq(\gamma-1)(\Delta-\widetilde{\Delta})\}\cap G_{k}=G_{k},$
we obtain \eqref{eq:FinalIneq1} from \eqref{eq:first-ineq}. 

Next, we show \eqref{eq:FinalIneq2}. First of all, we note that $G_{k}$
is independent of $\mathcal{H}_{k-1}^{i}$. Therefore, 
\begin{align}
\mathbb{P}[G_{k}|\mathcal{H}_{k-1}^{i}] & =\mathbb{P}[G_{k}].\label{eq:pgequality}
\end{align}
To compute $\mathbb{P}[G_{k}]$, we note that $t_{k\gamma}^{i},t_{k\gamma+1}^{i},\ldots,t_{(k+1)\gamma-2}^{i}$
are independent, and thus, the events $F_{k\gamma}$, $F_{k\gamma+1}$,
$\cdots$, $F_{(k+1)\gamma-2}$ are also independent. As a result,
\begin{align}
 & \mathbb{P}[G_{k}]=\mathbb{P}[F_{k\gamma}\cap F_{k\gamma+1}\cap\cdots\cap F_{(k+1)\gamma-2}]\nonumber \\
 & \,\,=\mathbb{P}[F_{k\gamma}]\cdots\mathbb{P}[F_{(k+1)\gamma-2}]=\Bigl(\cfrac{\widetilde{\Delta}}{\Delta}\Bigr)^{\gamma-1}=q^{\gamma-1}.\label{eq:P-Gk}
\end{align}
 Hence, \eqref{eq:FinalIneq2} follows from \eqref{eq:pgequality}
and \eqref{eq:P-Gk}. Finally, we use \eqref{eq:probability-trick-ahmet}--\eqref{eq:FinalIneq2}
to obtain \eqref{eq:last-prob}, leading us to \eqref{eq:Phi-last-prob-1}.
\hfill $\square$

\section{Proof of Lemma~\ref{LemmaIntersection}}
\label{Appendix-Proof-Intersection}

We show \eqref{eq:IntersectionLemmaResult} by induction. First, we
consider the case where $N=1$. In this case, we obtain
\begin{align}
 & \mathbb{P}[\cap_{k=0}^{N-1}\{\hat{\varphi}_{k}^{i}=\overline{\varphi}_{k+1}\}]\nonumber \\
 & \,\,=\mathbb{P}[\hat{\varphi}_{0}^{i}=\overline{\varphi}_{1}]=\mathbb{P}[\hat{\varphi}_{0}^{i}=1]\overline{\varphi}_{1}+\mathbb{P}[\hat{\varphi}_{0}^{i}=0](1-\overline{\varphi}_{1})\nonumber \\
 & \,\,\leq\overline{\varphi}_{1}+\mathbb{P}[\hat{\varphi}_{0}^{i}=0](1-\overline{\varphi}_{1}).\label{eq:IntersectionLemmaFirstCase}
\end{align}
By Lemma~\ref{Lemma:Positive-prob}, we have $\mathbb{P}[\hat{\varphi}_{0}^{i}=0|\mathcal{H}_{-1}^{i}]=1-\mathbb{P}[\hat{\varphi}_{0}^{i}=1|\mathcal{H}_{-1}^{i}]\leq1-2q^{\gamma}$.
Hence, $\mathbb{P}[\hat{\varphi}_{0}^{i}=0]=\mathbb{E}[\mathbb{P}[\hat{\varphi}_{0}^{i}=0|\mathcal{H}_{-1}^{i}]]=\mathbb{E}[1-\mathbb{P}[\hat{\varphi}_{0}^{i}=1|\mathcal{H}_{-1}^{i}]]\leq\mathbb{E}[1-2q^{\gamma}]=1-2q^{\gamma}$.
As a result, $\overline{\varphi}_{1}+\mathbb{P}[\hat{\varphi}_{0}^{i}=0](1-\overline{\varphi}_{1})\leq\overline{\varphi}_{1}+(1-2q^{\gamma})(1-\overline{\varphi}_{1})=1-2q^{\gamma}(1-\overline{\varphi}_{1})$.
Thus, we have \eqref{eq:IntersectionLemmaResult} for $N=1$.

Next, assume \eqref{eq:IntersectionLemmaResult} holds for $N=\tilde{N}>2$,
that is 
\begin{align}
\mathbb{P}[\bigcap_{k=0}^{\tilde{N}-1}\{\hat{\varphi}_{k}^{i}=\overline{\varphi}_{k+1}\}] & \leq\prod_{j=1}^{\tilde{N}}\big(1-2q^{\gamma}(1-\overline{\varphi}_{j})\big).\label{eq:tildenassumption}
\end{align}
 We will show that \eqref{eq:IntersectionLemmaResult} holds for $N=\tilde{N}+1$.
Observe
\begin{align}
 & \mathbb{P}\big[\bigcap_{k=0}^{\tilde{N}}\{\hat{\varphi}_{k}^{i}=\overline{\varphi}_{k+1}\}\big]=\mathbb{E}\big[\prod_{k=0}^{\tilde{N}}\mathbbm{1}[\hat{\varphi}_{k}^{i}=\overline{\varphi}_{k+1}]\big]\nonumber \\
 & \quad=\mathbb{E}\big[\prod_{k=0}^{\tilde{N}}\mathbbm{1}[\hat{\varphi}_{k}^{i}=\overline{\varphi}_{k+1}]|\mathcal{H}_{\tilde{N}-1}^{i}\big].\label{eq:peeq}
\end{align}
Since the random variables $\hat{\varphi}_{k}^{i}$, $k\in\{0,\ldots,\tilde{N}-1\}$,
are $\mathcal{H}_{\tilde{N}-1}^{i}$-measurable, from \eqref{eq:peeq},
we obtain 
\begin{align}
 & \mathbb{P}\big[\bigcap_{k=0}^{\tilde{N}}\{\hat{\varphi}_{k}^{i}=\overline{\varphi}_{k+1}\}\big]=\mathbb{E}\Big[\big(\prod_{k=0}^{\tilde{N}-1}\mathbbm{1}[\hat{\varphi}_{k}^{i}=\overline{\varphi}_{k+1}]\big)\nonumber \\
 & \quad\cdot\mathbb{E}\big[\mathbbm{1}[\hat{\varphi}_{\tilde{N}}^{i}=\overline{\varphi}_{\tilde{N}+1}]|\mathcal{H}_{\tilde{N}-1}^{i}\big]\Big].\label{eq:peeq2}
\end{align}
By Lemma~\ref{Lemma:Positive-prob}, we have $\mathbb{P}[\hat{\varphi}_{\tilde{N}}^{i}=0|\mathcal{H}_{\tilde{N}-1}^{i}]=1-\mathbb{P}[\hat{\varphi}_{\tilde{N}}^{i}=1|\mathcal{H}_{-1}^{i}]\leq1-2q^{\gamma}$.
Thus, 
\begin{align}
 & \mathbb{E}\left[\mathbbm{1}\left[\hat{\varphi}_{\tilde{N}}^{i}=\overline{\varphi}_{\tilde{N}+1}\right]|\mathcal{H}_{\tilde{N}-1}^{i}\right]=\mathbb{P}[\hat{\varphi}_{\tilde{N}}^{i}=\overline{\varphi}_{\tilde{N}+1}|\mathcal{H}_{\tilde{N}-1}^{i}]\nonumber \\
 & \,\,=\mathbb{P}[\hat{\varphi}_{\tilde{N}}^{i}=1|\mathcal{H}_{\tilde{N}-1}^{i}]\overline{\varphi}_{\tilde{N}+1}\nonumber \\
 & \,\,\quad+\mathbb{P}[\hat{\varphi}_{\tilde{N}}^{i}=0|\mathcal{H}_{\tilde{N}-1}^{i}](1-\overline{\varphi}_{\tilde{N}+1})\nonumber \\
 & \,\,\leq\overline{\varphi}_{\tilde{N}+1}+(1-2q^{\gamma})(1-\overline{\varphi}_{\tilde{N}+1})\nonumber \\
 & \,\,=1-2q^{\gamma}(1-\overline{\varphi}_{\tilde{N}+1}).\label{eq:eineq}
\end{align}
 It then follows from \eqref{eq:peeq2} and \eqref{eq:eineq} that
\begin{align}
 & \mathbb{P}\big[\bigcap_{k=0}^{\tilde{N}}\{\hat{\varphi}_{k}^{i}=\overline{\varphi}_{k+1}\}\big]\nonumber \\
 & \,\,\leq\mathbb{E}\big[\big(\prod_{k=0}^{\tilde{N}-1}\mathbbm{1}\left[\hat{\varphi}_{k}^{i}=\overline{\varphi}_{k+1}\right]\big)\big(1-2q^{\gamma}(1-\overline{\varphi}_{\tilde{N}+1})\big)\big]\nonumber \\
 & \,\,=\mathbb{E}\big[\prod_{k=0}^{\tilde{N}-1}\mathbbm{1}\left[\hat{\varphi}_{k}^{i}=\overline{\varphi}_{k+1}\right]\big]\big(1-2q^{\gamma}(1-\overline{\varphi}_{\tilde{N}+1})\big)\nonumber \\
 & \,\,=\mathbb{P}\big[\bigcap_{k=0}^{\tilde{N}-1}\{\hat{\varphi}_{k}^{i}=\overline{\varphi}_{k+1}\}\big]\big(1-2q^{\gamma}(1-\overline{\varphi}_{\tilde{N}+1})\big).\label{eq:peineq}
\end{align}
 Finally, by using \eqref{eq:tildenassumption} and \eqref{eq:peineq},
we obtain \eqref{eq:IntersectionLemmaResult} with $N=\tilde{N}+1$,
which completes the proof. \hfill $\square$
\end{document}